\documentclass[a4paper,english]{article}

\usepackage[margin=2.68cm]{geometry}

\usepackage{graphicx}
\usepackage{amsmath}
\usepackage{amsthm}
\usepackage{amssymb}
\usepackage{algorithm}
\usepackage{algorithmic}
\usepackage{calc}
\usepackage{mathtools}
\usepackage{ragged2e}
\usepackage{enumitem}

\usepackage{newfloat}

\DeclareFloatingEnvironment[fileext=frm,placement={!ht},name=Frame]{algfloat}

\newcounter{algono}
\newcommand{\algo}[1]{\refstepcounter{algono}\label{#1}}

\newtheorem{theorem}{Theorem}
\newtheorem{lemma}{Lemma}
\newtheorem{corollary}{Corollary}
\newtheorem{definition}{Definition}

\newlength\ubwidth

\newcommand{\Order}{\mathrm{O}}

\usepackage[hidelinks]{hyperref}
\newcommand{\mypara}[1]{\smallskip\noindent\textbf{#1.}}  

\newcommand{\Exp}{\mathbb{E}}

\renewcommand{\Pr}{\mathbb{P}}
\newcommand{\dist}{\text{dist}}

\date{}

\title{Distributed Minimum Vertex Coloring and Maximum Independent Set in Chordal Graphs\thanks{C.~Konrad carried out most work on this paper while being at the University of Warwick. He was supported by the Centre for Discrete Mathematics and its Applications (DIMAP) at Warwick University and by EPSRC award EP/N011163/1. V.~Zamaraev is supported by EPSRC award EP/P020372/1.}}

\author{Christian Konrad\thanks{Department of Computer Science, University of Bristol, UK}
\and
Viktor Zamaraev\thanks{Department of Computer Science, Durham University, UK}}

\begin{document}

\maketitle

\begin{abstract}
We give deterministic distributed $(1+\epsilon)$-approximation algorithms for Minimum Vertex Coloring and Maximum Independent Set on chordal graphs in the {\sf LOCAL} model. Our coloring algorithm runs in $\Order(\frac{1}{\epsilon} \log n)$ rounds, and our independent set algorithm has a runtime of $\Order(\frac{1}{\epsilon}\log(\frac{1}{\epsilon})\log^* n)$ rounds. For coloring, existing lower bounds imply that the dependencies on $\frac{1}{\epsilon}$ and $\log n$ are best possible. For independent set, we prove that $\Omega(\frac{1}{\epsilon})$ rounds are necessary.
  
  Both our algorithms make use of a tree decomposition of the input chordal graph. They iteratively peel off interval subgraphs, which are identified via the tree decomposition of the input graph, thereby partitioning the vertex set into $\Order(\log n)$ layers. For coloring, each interval graph is colored independently, which results in various coloring conflicts between the layers. These conflicts are then resolved in a separate phase, using the particular structure of our partitioning. For independent set, only the first $\Order( \log \frac{1}{\epsilon})$ layers are required as they already contain a large enough independent set. We develop a $(1+\epsilon)$-approximation maximum independent set algorithm for interval graphs, which we then apply to those layers. 

This work raises the question as to how useful tree decompositions are for distributed computing. 
\end{abstract}

 \section{Introduction} 
 \mypara{The {\sf LOCAL} Model}
 In the {\sf LOCAL} model of distributed computation, the input graph $G=(V, E)$ represents a communication network, where every 
 network node hosts a computational entity. Nodes have unique IDs. A distributed algorithm is executed on all network nodes simultaneously 
 and proceeds in discrete rounds. Initially, besides their own IDs, nodes only know their neighbors. Each round consists of a computation and a communication 
 phase. In the computation phase, nodes are allowed to perform unlimited computations. In the communication phase, nodes can send individual 
 messages of unbounded sizes to all their neighbors (and receive messages from them as well). The runtime of the algorithm is the total number of communication
 rounds, and the objective is to design algorithms that run in as few rounds as possible. The output is typically distributed: 
 For vertex colorings, it is required that upon termination of the algorithm, every node knows its own color, and for independent sets,
 every node knows whether it participates in the independent set.
 
 \mypara{Distributed Vertex Coloring}
 Vertex coloring problems have been studied in distributed computational models since more than 30 years (e.g. \cite{cv86,gps88}).  
 Given a graph $G=(V, E)$, a {\em (legal) $c$-coloring} of $G$ is an assignment $\gamma: V \rightarrow \{1, 2, \dots, c \}$ of at most $c$ colors to the 
 nodes of $G$ such that every pair of adjacent nodes receives different colors. The algorithmic challenge lies in computing colorings
 with few colors. The {\em chromatic number} $\chi(G)$ is the smallest $c$ such that there is a $c$-coloring of $G$. The \textsc{Minimum
 Vertex Coloring} problem (\textsc{MVC}) consists of finding a $\chi(G)$-coloring and is one of the problems studied in this paper. This is a difficult task, 
 even in the centralized setting: In general graphs, \textsc{MVC} is NP-complete \cite{k72} and hard to approximate within a factor of $n^{1-\epsilon}$, for every $\epsilon > 0$ \cite{z07}.
 
 Most research papers on distributed vertex coloring address the problem of 
 computing a $(\Delta+1)$-coloring, where $\Delta$ is the maximum degree of the input graph. Sequentially, a simple greedy algorithm that traverses 
 the nodes in arbitrary order and assigns the smallest color possible solves this problem. Distributively, this is a non-trivial task, and a long line 
 of research has culminated in the randomized algorithm of Harris et al. \cite{hss16}, which runs in 
 $\Order(\sqrt{\log \Delta} + 2^{\Order( \sqrt{ \log \log n } ) } )$ rounds,
 and the deterministic algorithm of Fraigniaud et al. \cite{fhk16}, which runs in $\Order(\sqrt{\Delta} \log^{2.5} \Delta + \log^* n)$ rounds.
 
 Only very few research papers address the \textsc{MVC} problem in a distributed model itself. On general graphs, the best 
 distributed algorithm computes a $\Order(\log n)$-approximation in $\Order(\log^2 n)$ rounds \cite{b12} and is based on
 the network decomposition of Linial and Saks \cite{ls93}. This algorithm uses exponential time computations, which due to the computational hardness of \textsc{MVC} is necessary unless $P=NP$. Barenboim et al. \cite{beg15} gave a 
 $\Order(n^{\epsilon})$-approximation algorithm that runs in $\exp \Order(1/\epsilon)$ rounds. Both the exponential time computations
 and the relatively large best known approximation factor of $\Order(\log n)$ on general graphs motivate the study of special graph classes. 
 Besides results on graph classes with bounded chromatic number (planar graphs \cite{gps88} and graphs of bounded arboricity 
 \cite{be10,gl17}), the only graph class with unbounded chromatic number that has been addressed in the literature are 
 {\em interval graphs}, which are the intersection graphs of intervals on the line. Halld\'{o}rsson and Konrad gave a $(1+\epsilon)$-approximation algorithm for 
 \textsc{MVC} on interval graphs that runs in $\Order(\frac{1}{\epsilon} \log^* n)$ rounds \cite{hk17} (see also their previous work \cite{hk14}). This work is the 
 most relevant related work to our results.

 \mypara{Distributed Independent Sets} An {\em independent set} in a graph $G=(V, E)$ is a subset of non-adjacent nodes 
 $I \subseteq V$. 
 Algorithms for independent sets are usually designed with one of the following two objectives in mind: 
 (1) Compute a {\em maximal independent set}, i.e., an independent set $I$ that cannot be enlarged by adding a node outside $I$ to it, or 
 (2) Compute a \textsc{Maximum Independent Set} (\textsc{MIS}) (or an approximation thereof), i.e., an independent set of maximum size, 
 which is the variant studied in this paper. Similar to \textsc{MVC}, the \textsc{MIS} problem is NP-complete \cite{k72} and hard to 
 approximate within a factor of $n^{1-\epsilon}$, for every $\epsilon > 0$ \cite{z07}. In the distributed setting, Luby \cite{l86} 
 and independently Alon et al. \cite{abi86} gave distributed $\Order(\log n)$ rounds maximal independent set algorithms more than 30 
 years ago. Improved results are possible for graphs with bounded maximum degree (\cite{beps16,g16}) or on specific graph 
 classes (e.g. \cite{cv86,sw10}). Using exponential time computations, a $(1+\epsilon)$-approximation to \textsc{MIS} can be computed 
 in general graphs in $\Order(\frac{1}{\epsilon} \log n)$ rounds \cite{bhkk16} (see also \cite{gkm17}). 
 Deterministic distributed \textsc{MIS} algorithms may be inferior to randomized ones: It is known that every deterministic \textsc{MIS}
 $\Order(1)$-approximation algorithm on a path requires $\Omega(\log^* n)$ rounds \cite{lw08,chw08}, while a simple 
 randomized $\Order(1)$ rounds $\Order(1)$-approximation algorithm exists \cite{chw08}.

 \mypara{Chordal Graphs}
 In this paper, we study \textsc{MVC} and \textsc{MIS} on {\em chordal graphs}. A graph is chordal, if every cycle on at least four 
 nodes contains a {\em chord}, i.e., an edge different from the edges of the cycle connecting two nodes of the cycle. Chordal 
 graphs play an important role in graph theory and have many applications, for example in belief propagation in machine learning 
 (e.g. the Junction Tree algorithm). They constitute a superclass of interval graphs and trees and an inportant subclass of perfect graphs. 
 The key motivations for our work are as follows: 
 
  \textbf{1. Minimum Vertex Colorings.} 
  Since the best known distributed \textsc{MVC} algorithm only gives a $\Order(\log n)$-approximation, we are interested in pinpointing graph
  structures that are difficult to handle. In this paper, we show that \textsc{MVC} and \textsc{MIS} can both be well solved on chordal graphs. A defining
  feature of a chordal graph is that it does not contain any induced cycles of lengths at least $4$. This in turn implies that difficult instances
  for distributed coloring necessarily contain induced cycles of length at least 4.
  
  Furthermore, as previously mentioned, \textsc{MVC} can be solved well on interval graphs in the distributed setting \cite{hk17}. 
  We are therefore interested in identifying more general graph classes that admit distributed $(1+\epsilon)$-approximation algorithms 
  for \textsc{MVC}. Since trees are chordal, Linial's lower bound for coloring trees applies \cite{l92}. Linial proved 
  that coloring trees with a constant number of colors requires 
  $\Omega(\log n)$ rounds, which gives a $\Omega(\log n)$ lower bound for any constant factor approximation to $\textsc{MVC}$ on chordal graphs.
  This separates the difficulties of \textsc{MVC} on chordal and interval graphs. Furthermore, Halld\'{o}rsson
  and Konrad proved that a $(1+\epsilon)$-approximation to \textsc{MVC} on interval graphs requires $\Omega(\frac{1}{\epsilon})$ rounds \cite{hk17}. 
  Combined, we obtain a $\Omega(\frac{1}{\epsilon} + \log n)$ lower bound on the round complexity for approximating \textsc{MVC} on chordal graphs within a factor of $1 + \epsilon$.
  
  
  \textbf{2. Tree Decompositions.} Tree decompositions are a powerful algorithmic tool that allow for the design of (centralized) linear time algorithms 
  for NP-hard problems on graphs of bounded tree-width (see below for precise definitions) \cite{arnborg1989linear}. They have played however only a minor role in 
  the design of distributed algorithms (few exceptions are \cite{gw10,asd12,nm16}). A {\em tree decomposition} of a graph $G = (V, E)$ is identified by 
  a set of bags $S_1, S_2, \dots \subseteq V$ that are arranged in a tree $\mathcal{T}$ such that every adjacent pair of nodes $uv \in E$ is contained in at 
  least one bag, and, for any $v \in V$, the set of bags that contain $v$ induces a subtree in $\mathcal{T}$. One potential reason for the limited 
  success of tree decompositions as a tool for distributed computing is that even simple graphs, such 
  as a ring on $n$ nodes, require that many bags of their tree decompositions consist of nodes that are at distance $\Omega(n)$ in the original 
  graph. For these graphs, it is thus impossible that nodes obtain coherent local views of a global tree decomposition in $o(n)$ rounds.
  
  Chordal graphs are well-suited for studying distributed algorithms that exploit the input graph's tree decomposition, since in a chordal graph $G$, each bag consists of a clique in $G$. Thus, every bag that contains a node $v \in V$ further only contains nodes that lie in $v$'s neighborhood. 
We exploit this locality property and show that in the {\sf LOCAL} model, nodes can indeed obtain coherent local views of a 
global tree decomposition. 
   

 \mypara{Results} In this paper, we give deterministic distributed $(1+\epsilon)$-approximation algorithms for 
 \textsc{MVC} and \textsc{MIS} on chordal graphs in the {\sf LOCAL} model. Our algorithm for \textsc{MVC} runs 
 in $\Order(\frac{1}{\epsilon} \log n)$ rounds (\textbf{Theorem~\ref{thm:dist-alg}}), and our algorithm for \textsc{MIS}
 has a runtime of $\Order(\frac{1}{\epsilon} \log(\frac{1}{\epsilon})  \log^* n )$ rounds (\textbf{Theorem~\ref{thm:dist-alg-mis}}). 
 For \textsc{MVC}, as mentioned above, the dependencies of the runtime on $\log n$ and $\frac{1}{\epsilon}$ are best 
 possible (though the existance of an algorithm with runtine $\Order(\frac{1}{\epsilon} + \log n)$ is not ruled out). 
 For \textsc{MIS}, we prove that every possibly randomized $(1+\epsilon)$-approximation algorithm requires $\Omega(\frac{1}{\epsilon})$ rounds, even on paths (\textbf{Theorem~\ref{thm:lb-mis}}). 
 
 \mypara{Techniques} Our algorithms for \textsc{MVC} and \textsc{MIS} first compute the tree decomposition of the input 
 chordal graph. Every chordal graph $G=(V, E)$ can be represented as the intersection graph of subtrees of a tree 
 $\mathcal{T}$. In this representation, every node $v \in V$ is identified with a subtree 
 $\mathcal{T}(v) \subseteq \mathcal{T}$ such that for nodes $u,v \in V$, subtrees $\mathcal{T}(u)$ and $\mathcal{T}(v)$ intersect if 
 and only if $uv \in E$. 
 
 We then distributively run the following peeling process: In each iteration $i$, we first identify the family 
 $\mathcal{L}_i$ of all pendant (i.e., incident to at least one leaf) and all {\em long enough} paths in the tree decomposition 
 (for an appropriate notion of 'long enough'), such that every path consists of vertices of degree at most $2$. We then peel 
 off those nodes from the current graph whose corresponding subtrees in the tree 
 decomposition are subpaths of the paths in $\mathcal{L}_i$. The removed nodes $V_i$ define layer $i$. We prove that after 
 $\Order(\log n)$ iterations of the peeling process, all nodes are assigned into layers. This partitioning has the useful 
 property that every layer $V_i$ induces an interval graph in $G$. We can thus use distributed algorithms designed for interval 
 graphs on those layers.
 
 For \textsc{MVC}, we use the algorithm of Halld\'{o}rsson and Konrad \cite{hk17} to color these layers individually and independently in 
 $\Order(\frac{1}{\epsilon} \log^* n)$ rounds, which results in various coloring conflicts between the layers. These are 
 then resolved in a separate phase, using the particular structure of the layers.
 
 For \textsc{MIS}, we only compute the first $\Order(\log \frac{1}{\epsilon})$ layers of the partitioning, since they already contain
 a large enough independent set for a $(1+\epsilon)$-approximation. In order to compute a large independent set in each layer, we first
 design a $(1+\epsilon)$-approximation $\Order(\frac{1}{\epsilon} \log^* n)$ rounds algorithm for $\textsc{MIS}$ on interval graphs, 
 which is then executed on these layers.

 \mypara{Outline} We proceed as follows. In Section~\ref{sec:prelim}, we give notation and definitions. We 
 then discuss in Section~\ref{sec:local-views} how network nodes can obtain coherent local views of the tree decomposition. A centralized 
 $(1+\epsilon)$-approximation \textsc{MVC} algorithm is then presented in Section~\ref{sec:coloring-centralized}, and a distributed implementation of this algorithm is given in Section~\ref{sec:coloring-distributed}. 
In Section \ref{sec:MISinterval}, we provide a distributed $(1+\epsilon)$-approximation \textsc{MIS} algorithm for interval graphs,
which is then employed in Section \ref{sec:mis-chordal} to design a distributed $(1+\epsilon)$-approximation \textsc{MIS} algorithm for chordal graphs. We complement the latter result in Section \ref{app:lowerMIS} by proving that $\Omega(\frac{1}{\epsilon})$ rounds are required for computing a $(1+\epsilon)$-approximation to \textsc{MIS}, even on paths.
Finally, we conclude in Section~\ref{sec:conclusion}.
 

\section{Preliminaries} \label{sec:prelim}
\mypara{Notation and Definitions for Graphs}
Let $G=(V, E)$ be a graph. For a node $v \in V$, we denote by $\Gamma_G(v)$ the \textit{neighborhood} of $v$ in $G$. 
The \textit{degree} of $v$ in $G$ is defined as $\deg_G(v) := |\Gamma_G(v)|$. 
By $\Gamma_G[v]$ we denote the set $\Gamma_G(v) \cup \{ v \}$. Similarly, for a set of nodes
$W \subseteq V$ we write $\Gamma_G(W) := (\bigcup_{v \in W} \Gamma_G(v) ) \setminus W$, and $\Gamma_G[W] := \Gamma_G(W) \cup W$. 
The \textit{distance-$k$ neighborhood} of $v$ in $G$, i.e., the set of nodes at distance at most $k$ from $v$ in $G$,
is denoted $\Gamma^{k}_G(v)$. The subgraph of $G$ induced by a set of nodes $U$ is denoted by $G[U]$.
A set of pairwise adjacent (resp., non-adjacent) nodes in $G$ is called a {\em clique} (resp., an {\em independent set}). 
A clique (resp., an indepedent set) $S$ is {\em maximal} if $S \cup \{v\}$ is not a clique (resp., an indepedent set), 
for every $v \in V \setminus S$.
A {\em maximum} independent set in $G$ is an independent set of maximum size. The cardinality of a maximum independent set is
called the {\em independence number} of $G$ and denoted by $\alpha(G)$.
When the context is clear, instead of $\alpha(G[S])$ we will simply write $\alpha(S)$.
A graph is \textit{chordal} if every cycle of length at least four contains a chord, i.e., an edge that connects two
non-consecutive nodes of the cycle.
It is a well-known fact that an $n$-node chordal graph has at most $n$ maximal cliques.


\mypara{Tree Decomposition}
A \textit{tree decomposition} of an $n$-node graph $G=(V, E)$ is a forest $\mathcal{T} = (\mathcal{S}, \mathcal{E})$ whose
vertex set $\mathcal{S} = \{ S_1, S_2, \ldots, S_n \}$ is a family of subsets of $V$, and:
\begin{enumerate}
	\item every node\footnote{For convenience, throughout the paper, we say \textit{node} when referring to a vertex of an 
	underlying graph, and we say \textit{vertex} when referring to a vertex of its tree decomposition.} 
	$v \in V$ belongs to at least one subset in $\mathcal{S}$;
	\item for every edge $uv \in E$, there is a subset $S_i \in \mathcal{S}$ containing both nodes $u$ and $v$;
	\item for every node $v \in V$ the family  $\phi(\mathcal{T}, v) \subseteq \mathcal{S}$ of subsets containing $v$ induces 
	a tree in $\mathcal{T}$, which we denote $\mathcal{T}(v)$, i.e., $\mathcal{T}(v) := \mathcal{T}[\phi(\mathcal{T}, v)]$.
\end{enumerate}

\noindent
When the tree decomposition is clear from the context we will write $\phi(v)$ instead of $\phi(\mathcal{T}, v)$. 
It follows from the definition that $G$ is a subgraph of the intersection graph of the trees $\mathcal{T}(v)$.

\mypara{Tree Decomposition of Chordal Graphs}
It is well known (see, e.g., \cite{blair1993introduction}) that a graph $G$ is chordal if and only if it has a tree decomposition 
$\mathcal{T} = (\mathcal{C}, \mathcal{E})$ whose vertex set $\mathcal{C}$ is the family of maximal cliques of $G$.
We call such a tree decomposition a \textit{clique forest} of chordal graph $G$.
Since every vertex of the clique forest is a clique, $G$ coincides with the intersection graph of the subtrees $\mathcal{T}(v)$ of 
clique forest $\mathcal{T}$. In other words, a clique forest of a chordal graph $G$ is a forest $\mathcal{T} = (\mathcal{C}, \mathcal{E})$ 
whose vertex set $\mathcal{C}$ is the family of maximal cliques of $G$, such that $\mathcal{T}[\phi(v)]$ is a tree for every $v$.
If a clique forest of a chordal graph is linear, i.e., a forest with every component being a path, then
the graph is interval.

\begin{theorem}[\cite{fishburn1985interval}]\label{th:interval-char}
	A chordal graph $G$ is interval if and only if its clique forest is a linear forest.
\end{theorem}

\mypara{Binary Paths} 
We say that a path $v_1, \dots, v_k$ in $G$ is {\em binary}, if $\deg_G(v_i) \le 2$,
 for every $i \in [k]$ (note that this implies that $\deg_G(v_i) = 2$, for every $i \in \{ 2, 3, \ldots, k-1 \}$). 
We say that a binary path $v_1, \dots, v_k$ is a {\em pendant path}, if either $\deg_G(v_1) = 1$ or $\deg_G(v_k) = 1$ (or both). 
For convenience, we consider an isolated vertex as a pendant path.
A binary path $v_1, \dots, v_k$ is an {\em internal path}, if $\deg_G(v_i) = 2$, for every $i \in [k]$. 
A binary/pendant/internal path is \textit{maximal} if it cannot be enlarged by adding a vertex outside the path to it.

Let $G=(V,E)$ be a chordal graph with clique forest $\mathcal{T} = (\mathcal{C}, \mathcal{E})$. Let $\mathcal{P} = C_1, \ldots, C_k$
be a binary path in $\mathcal{T}$. We define the \textit{diameter} of $\mathcal{P}$ to be the maximum distance in $G$ between nodes in $C_1 \cup \ldots \cup C_k$, that is,
$diam(\mathcal{P}) = \max\limits_{u \in C_i, v \in C_j, i,j \in [k]} dist_G(u,v).$
Similarly, we define the \textit{independence number} of $\mathcal{P}$ as $\alpha(G[C_1 \cup \ldots \cup C_k])$.

\mypara{Distributed Algorithms for Interval Graph} 
Halld\'{o}rsson and Konrad \cite{hk17} gave a deterministic distributed algorithm for coloring interval graphs. For every $\epsilon \ge \frac{2}{\chi(G)}$, 
their algorithm computes a $(1+\epsilon)$-approximation to $\textsc{MVC}$ in $\Order(\frac{1}{\epsilon} \log^* n)$ rounds. We will reuse this algorithm and denote it by $\textsc{ColIntGraph}(\epsilon)$.

 \section{Computing Local Views of the Clique Forest} \label{sec:local-views}
Our algorithms make use of the clique forest of the input chordal graph. 
For network nodes to obtain a coherent view of the clique forest, we make use of the following 
\textit{maximum weight spanning forest} characterization: With a chordal graph $G$ we associate the 
\textit{weighted clique intersection graph} $\mathcal{W}_G$ whose vertex set is the family $\mathcal{C}$ of maximal 
cliques of $G$, and any two cliques $C_1, C_2 \in \mathcal{C}$ with a nonempty intersection are connected by an edge 
with weight $|C_1 \cap C_2|$. Then:

\begin{theorem}[\cite{bernstein1981power}]
	A forest $\mathcal{T} = (\mathcal{C}, \mathcal{E})$ is a clique forest of a chordal graph $G$ if and only if it 
	is a maximum weight spanning forest of $\mathcal{W}_G$.
\end{theorem}

Observe that while the vertex set of a clique forest is unique, i.e., the family of maximal cliques of $G$, the edge 
set is not necessarily unique as there may be multiple different maximum weight spanning forests in $\mathcal{W}_G$.
To obtain coherent local views of a clique forest, it is thus necessary that nodes agree on a 
unique maximum weight spanning forest in $\mathcal{W}_G$. This can be achieved by defining a linear order $<$ on the edges 
of $\mathcal{W}_G$ that respects the partial order given by the edge weights, 
and preferring edges that are larger with respect to $<$. 
To this end, we first assign to every maximal clique $C \in \mathcal{C}$ a word $\sigma(C)$ over the 
alphabet of the identifiers of nodes, where $\sigma(C)$ consists of the identifiers of the nodes in $C$ listed in increasing order. 
Further, we associate with every edge $e = C_i C_j$ a triple $(w_e, l_e, h_e)$, 
where $w_e$ is the weight of $e$, i.e., $w_e = |C_i \cap C_j|$, $l_e = \text{lexmin} \{ \sigma(C_i), \sigma(C_j) \}$, and 
$h_e = \text{lexmax} \{ \sigma(C_i), \sigma(C_j) \}$. Now for two edges $e$ and $f$ we define $e < f$ if and only if
either $w_e < w_f$, or $w_e = w_f$ and $l_e \prec l_f$, or $w_e = w_f$, $l_e = l_f$ and $h_e \prec h_f$, where $\prec$ is the lexicographical order. 
Clearly, $<$ orders the edges of $\mathcal{W}_G$ linearly while preserving the weight order.

In what follows, when we say that $\mathcal{T}$ is \textit{the} clique forest of a chordal graph $G$, we implicitly assume that 
it is the clique forest uniquely specified by the mechanism. Figure \ref{fig:W_G_CliqueForest} demonstrates the weighted clique 
intersection graph and the clique forest of the chordal graph presented in Figure \ref{fig:graphG}.

\begin{figure}
	\centering
         \includegraphics[scale=1.3]{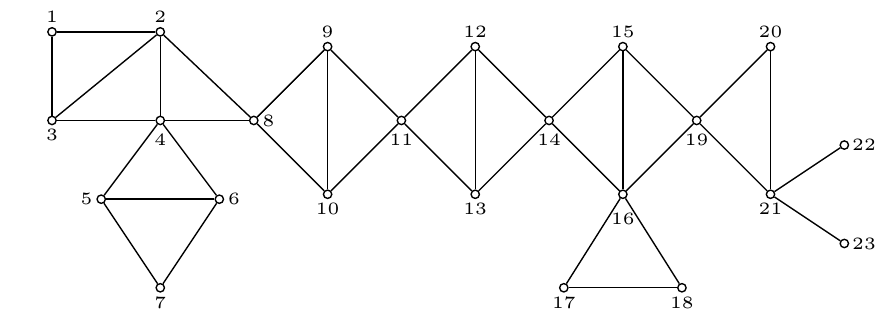}
         \caption{\small{A chordal graph $G$.}}
         \label{fig:graphG}

\vspace{4ex}

	\centering
         \includegraphics[scale=0.79]{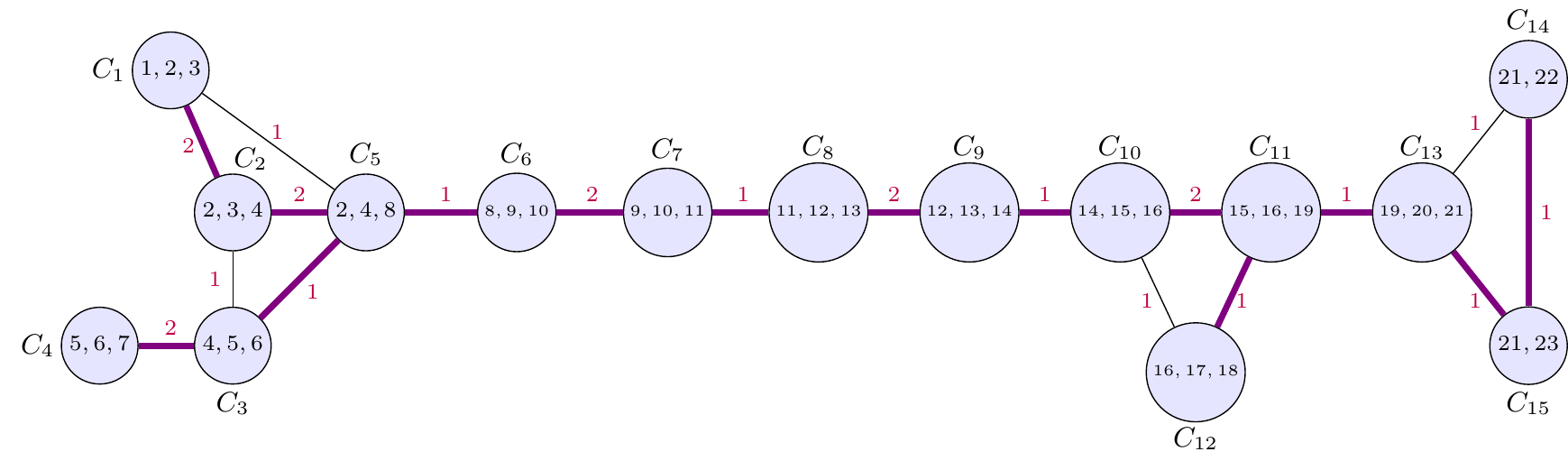}
         \caption{\small{The weighted clique intersection graph $\mathcal{W}_G$ of chordal graph $G$ presented in Fig. \ref{fig:graphG}.
         The vertices of $\mathcal{W}_G$ are the maximal cliques of $G$, 
         and two vertices $C_i, C_j$ of $\mathcal{W}_G$ with a nonempty intersection are connected 
         by an edge with weight $|C_1 \cap C_2|$.
         The bold edges are the edges of the clique forest $\mathcal{T}$ of $G$., i.e., the edges of the unique
         maximum weight spanning forest of $\mathcal{W}_G$ corresponding to the linear order of edges $<$.}}
         \label{fig:W_G_CliqueForest}

\vspace{4ex}

	\centering
         \includegraphics[scale=1.3]{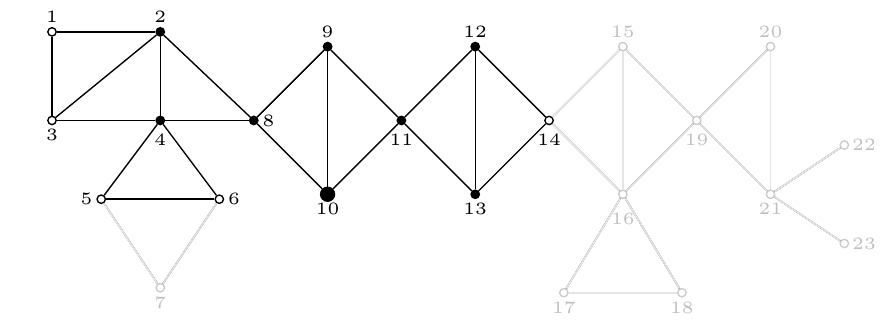}
         \caption{\small{Local view of graph $G$ from node $10$. The non-gray nodes are the nodes in $\Gamma_G^3[10]$, 
         and the black nodes are the nodes in $\Gamma_G^2[10]$.}}
         \label{fig:localView}

\vspace{4ex}

	\centering
         \includegraphics[scale=0.9]{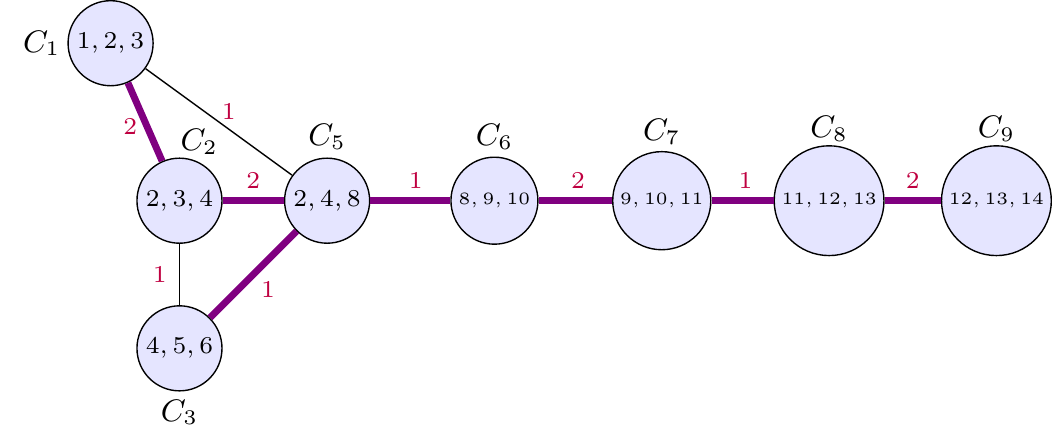}
         \caption{\small{Local view of the graph $\mathcal{W}_G$ from node $10$. 
         The cliques in $\mathcal{C}' = \{ C_1, C_2, C_3, C_5, C_6, C_7, C_8, C_9 \}$ are exactly the maximal cliques of $G$ that contain
         at least one node from $\Gamma_G^2[10]$.
         The bold edges are the edges of the unique maximum weight spanning forest of $\mathcal{W}_G[\mathcal{C}']$,
         which coincides with the subtree of $\mathcal{T}$ induced by $\mathcal{C}'$.}}
         \label{fig:W_G_CliqueForest_local}
\end{figure}

A maximum weight spanning forest has the following local optimality property: 
\begin{lemma}\label{cl:unique-subtree}
	Let $G = (V,E)$ be a weighted graph with a unique maximum weight spanning forest $F$, and let $U \subseteq V$
	be a set of nodes inducing a tree $T$ in $F$. Then $G[U]$ has a unique maximum weight spanning tree, which
	coincides with $T$.
\end{lemma}
Applied to a chordal graph and its clique forest, we thus obtain:
\begin{lemma} \label{lem:local-view-new}
	Let $G=(V, E)$ be a chordal graph and $\mathcal{T} = (\mathcal{C}, \mathcal{E})$ its clique forest. 
	Then for every node $v \in V$ the unique maximum weight spanning forest in $\mathcal{W}_G[\phi(v)]$ 
	equals to tree $\mathcal{T}(v) = \mathcal{T}[\phi(v)]$.
\end{lemma}
This suggests a method for a node $v \in V$ to compute a local view $\mathcal{T}'$ of clique forest 
$\mathcal{T}$: Suppose that $v$ knows its distance $d$-neighborhood $\Gamma_G^d[v]$.
For every $u \in \Gamma_G^{d-1}[v]$, $v$ computes the family $\phi(u)$ of maximal cliques containing $u$ 
(notice that a maximal clique that contains a node at distance $d-1$ from $v$ 
may include nodes at distance $d$). Then, $v$ computes the maximum weight spanning forest in every 
$\mathcal{W}_G[\phi(u)]$ and adds the edges of this forest to $\mathcal{T}'$. 
Figures \ref{fig:localView} and \ref{fig:W_G_CliqueForest_local} illustrate construction of a local view
of the clique forest of the chordal graph presented in Fig. \ref{fig:graphG}.

\section{Minimum Vertex Coloring: Centralized Algorithm}\label{sec:coloring-centralized}
 In this section, we give a centralized $(1+\epsilon)$-approximation algorithm for \textsc{MVC} on chordal graphs. This algorithm
 will later be implemented in the {\sf LOCAL} model in Section~\ref{sec:coloring-distributed}.
\subsection{Algorithm}
 
 \begin{algfloat}
 \normalsize
 \noindent\fbox{\parbox{\textwidth - 7pt}{

 \textbf{Input:} $G=(V, E)$ is an $n$-node chordal graph with clique forest $\mathcal{T} = (\mathcal{C}, \mathcal{E})$;
 a parameter $\epsilon > \frac{2}{\chi(G)}$.

 \vspace{-0.2cm}
 
 \begin{enumerate}
  \item[] Set $k = 2/\epsilon$.
  \item \textbf{Pruning Phase.} \\
Let $\mathcal{T}_1 = \mathcal{T}$, $U_1 = V$.\\
\textbf{for} $i=1,2, \dots, \lceil \log n \rceil$ \textbf{do}:
	\begin{enumerate}[topsep=0cm,leftmargin=1cm]
  		\item Let $\mathcal{L}_i$ be the set that contains all maximal pendant paths of $\mathcal{T}_i$, 
			and all maximal internal paths of $\mathcal{T}_i$ of diameter at least $3k$.
		\item Let $V_i \subseteq U_i$ be such that for each $v \in V_i$, $\mathcal{T}(v)$ is a 
			subpath of one of the paths in $\mathcal{L}_i$. 
		\item Let $U_{i+1} = U_i \setminus V_i$, and let $\mathcal{T}_{i+1}$ be the forest obtained 
			from $\mathcal{T}_i$ by removing all paths in $\mathcal{L}_i$. 
			As proved in Lemma~\ref{lem:tree-decomp}, $\mathcal{T}_{i+1}$ is the clique forest of $G[U_{i+1}]$.
	\end{enumerate}
  
  \vspace{0.2cm}
  
\item \textbf{Coloring Phase.} \\
\textbf{for} $i=1,2, \dots, \lceil \log n \rceil$ \textbf{do}: 
\begin{enumerate}[topsep=0cm,leftmargin=0.5cm]
	\item[] Color $G[V_i]$ with at most $(1+1/k) \chi(G[V_i]) + 1$ colors. 
\end{enumerate} 
 
  \vspace{0.2cm}  
  \item \textbf{Color Correction Phase.} \\
\textbf{for} $i = \lceil \log n \rceil - 1, \lceil \log n \rceil - 2, \dots, 1$ \textbf{do}:
\begin{enumerate}[topsep=0cm,leftmargin=0.5cm]
\item[]
\textbf{for} each path $\mathcal{P} \in \mathcal{L}_i$ \textbf{do}: 
\begin{enumerate}[leftmargin=1cm]
	\item[(a)] Let $W \subseteq V_i$ be the set of nodes $w$ such that $\mathcal{T}(w)$ is a subpath of $\mathcal{P}$. 
	\item[(b)] Let $W' \subseteq \bigcup_{l > i} V_l$ be the subset of nodes that have neighbors in $W$. 
	\item[(c)] By Lemma~\ref{lem:conflict-vertices} $G[W \cup W']$ is an interval graph. 
		Applying Lemma~\ref{lem:recoloring}, we recolor those nodes of $W$ 
		that are at distance at most $k+3$ from some node in $W'$ using at most $(1+1/k) \chi(G[V_i]) + 1$ 
		colors to resolve all coloring conflicts between $W$ and $W'$.
\end{enumerate}
 \end{enumerate}
\end{enumerate}
\vspace{-0.3cm}
}}
\begin{center}
 \algo{alg:coloringChordal}
 \textbf{Algorithm \ref*{alg:coloringChordal}.} A centralized $(1 + \epsilon)$-approximation coloring algorithm for chordal graphs.
\end{center}
\vspace{-0.5cm}
\end{algfloat}

Our algorithm (Algorithm~\ref{alg:coloringChordal}) consists of the pruning, the coloring, and the color correction phase:
 
 In the pruning phase, the node set $V$ is partitioned into at most $\lceil \log n \rceil$ \textit{layers} $V_1, \dots, V_{\lceil \log n \rceil}$ 
 such that, for every $i \in [\lceil \log n \rceil]$, $G[V_i]$ constitutes an interval graph.  
 In each step of the pruning phase, we remove every node $v \in U_i$ from the current graph $G[U_i]$ (we set $U_1 = V$ and 
 hence $G[U_1] = G$) whose 
 corresponding subtree $\mathcal{T}(v)$ in the clique forest $\mathcal{T}_i$ of $G[U_i]$ is a subpath of a pendant path 
 or an internal path of diameter at least $3k$. 
The set of removed nodes is denoted $V_i$, and $G[V_i]$
 forms an interval graph (which follows from Lemma~\ref{lem:helpful-1}). 
 We prove in Lemma~\ref{lem:tree-decomp} that the clique forest $\mathcal{T}_{i+1}$
 of the resulting graph $G[U_{i+1}]$, where $U_{i+1} = U_i \setminus V_i$, can be obtained by removing all pendant paths and all internal paths of diameter at least $3k$ from $\mathcal{T}_i$. We also show in Lemma~\ref{lem:pruning} that the pruning process ends after at most $\lceil \log n \rceil$ iterations and thus creates at most $\lceil \log n \rceil$ layers.
 
 In the coloring phase, each interval graph $G[V_i]$ is colored with at most $(1+1/k) \chi(G[V_i]) + 1$ colors. 
 In the centralized setting, it would be easy to color these interval graphs optimally. 
However, since we will implement this algorithm later in the distributed setting, and an optimal coloring on interval graphs cannot be
computed distributively in few rounds, we impose a weaker quality guarantee that can be achieved distributively. 
The colorings of different layers are 
 computed independently from each other and do not give a coherent coloring of the entire input graph.

 In the color correction phase, these incoherencies are corrected. To this end, the colors of $V_{\lceil \log n \rceil}$ remain unchanged and we correct
 the layers iteratively, starting with layer $\lceil \log n \rceil - 1$ and proceeding downwards to layer $V_1$. In a general step, for
 every path $\mathcal{P} \in \mathcal{L}_i$, we show that the subgraph induced by the nodes $W \subseteq V_i$ whose subtrees are subpaths of $\mathcal{P}$ forms 
 an interval graph together with those nodes in $\bigcup_{j \ge i+1} V_j$ that have coloring conflicts towards $W$ (Lemma~\ref{lem:conflict-vertices}). 
Notice that each path $\mathcal{P}$ connects to at most two maximal cliques in $\mathcal{T}_i$. The neighborhood of $W$ thus consists 
 of subsets of these (at most two) cliques, which further implies that all conflicting nodes in $\bigcup_{j \ge i+1} V_j$
 are included in these cliques as well. We then reuse a recoloring result previously proved by Halld\'{o}rsson and Konrad 
 \cite{hk17}, which shows that we can resolve all conflicts by changing the colors of those nodes in $W$ that are 
at distance at most $k+3$ from the (at most) two conflicting cliques.



\subsection{Analysis} In the first part of our analysis, we show that throughout the algorithm, 
$\mathcal{T}_{i+1}$ as computed in Step 1(c) is the (unique) tree decomposition of $G[U_{i+1}]$. 
Recall that $U_{i+1}$ is obtained from $U_i$ by removing nodes whose corresponding subtrees 
in the tree decomposition $\mathcal{T}_i$ are contained in pendant and internal paths. We prove
that it is enough to remove the pendant and internal paths from $\mathcal{T}_i$ in order to obtain
$\mathcal{T}_{i+1}$. This is an important property as it allows us to bound the number of iterations
required to partition all nodes into layers. We first address internal paths in 
Lemma~\ref{lem:internal-path-removal} (see Figures~\ref{fig:graphG-U} and \ref{fig:cliqueForestOfG-U}
for an illustration), and then state a similar result for pendant paths in Lemma~\ref{lem:pendant-path-removal}.

\begin{lemma}\label{lem:internal-path-removal}
	Let $\mathcal{T} = (\mathcal{C}, \mathcal{E})$ be the clique forest of a chordal graph $G=(V, E)$.
	Let $\mathcal{P} = C_1, C_2, \dots, C_k$ be an internal path in $\mathcal{T}$, which is connected to $\mathcal{T}$
	by edges $C_sC_1$ and $C_kC_e$.
	If $diam(\mathcal{P}) \geq 4$, then
	$\mathcal{T} - \mathcal{P}$ is the clique forest of $G[V \setminus U]$, where 
	$U = \{u \in V \, : \, \mathcal{T}(u) \mbox{ is a subpath of } \mathcal{P} \}$.
\end{lemma}
\begin{proof}
	Notice that the condition $diam(\mathcal{P}) \geq 4$ implies that no node in $C_s$ is adjacent to a node in $C_e$,
	and, in particular, $C_s \cap C_e = \emptyset$.
	
	It is enough to show that 
	$\mathcal{T}' = \mathcal{T} - \mathcal{P}$ is a clique forest of $G[V \setminus U]$. The uniqueness of $\mathcal{T}'$
	then follows from the uniqueness of $\mathcal{T}$.
	First, we show that $\mathcal{C} \setminus \{ C_1, C_2, \dots, C_k \}$ is the family of maximal cliques of $G[V \setminus U]$.
	Indeed, by construction $U \subseteq (C_1 \cup \ldots \cup C_k) \setminus (C_s \cup C_e)$, 
	and since no node in $C_s$ is adjacent to a node in $C_e$,
	we have that for every $i \in[k]$ either $C_i \subseteq C_s \cup U$, or $C_i \subseteq C_e \cup U$. Therefore, by removing
	from $G$ the nodes in $U$ we destroy the maximal cliques $C_1, C_2, \dots, C_k$, and do not affect the others.
	It remains to show that $\mathcal{T}'(v)$ is a tree for every $v \in V \setminus U$. Assume to the contrary, that 
	$\mathcal{T}'(v)$ is disconnected for some $v \in V \setminus U$. Since $\mathcal{T}(v)$ is a tree, and 
	$\mathcal{T}' (v)= \mathcal{T}(v) - \mathcal{P}$, we conclude that $v$ belongs to both $C_s$ and $C_e$.
	But this contradicts the disjointness of $C_s$ and $C_e$. Hence the lemma.
\end{proof}


\begin{figure}
	\centering
         \includegraphics[scale=1.3]{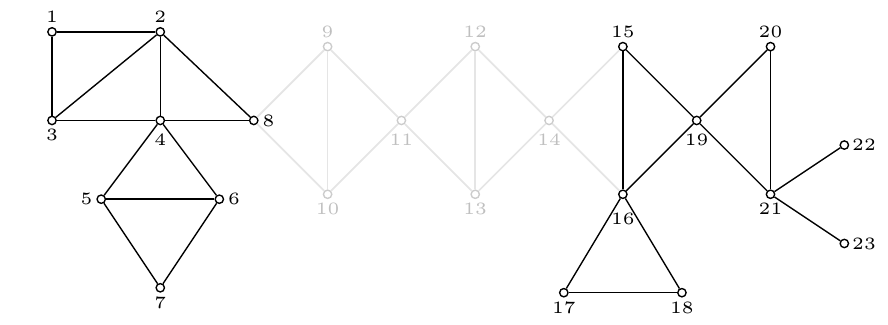}
         \caption{\small{The black nodes represent the subgraph $G[V \setminus U]$ of $G$, 
         where $U$ is the set of nodes $u$ whose corresponding trees $\mathcal{T}(u)$ are subpaths of
         path $\mathcal{P} = C_6,C_7,C_8,C_9,C_{10}$ in the clique forest $\mathcal{T}$ of graph $G$.}}
         \label{fig:graphG-U}

\vspace{4ex}

	\centering
         \includegraphics[scale=0.79]{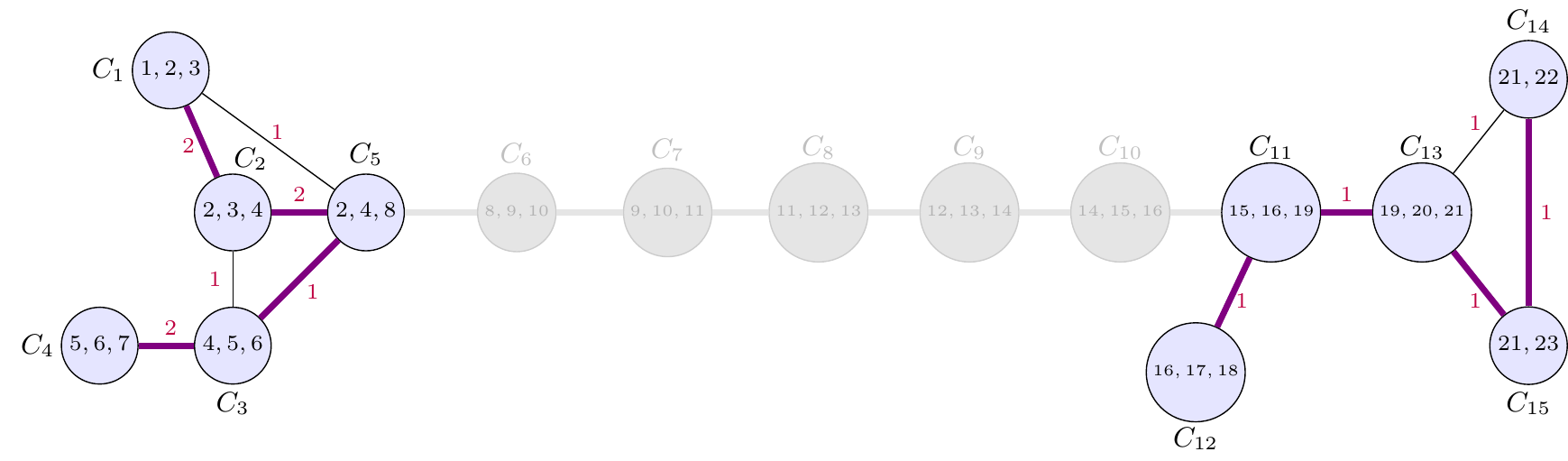}
         \caption{\small{The non-gray vertices represent the clique forest of $G[V \setminus U]$, which
         coincides with $\mathcal{T} - \mathcal{P}$.}}
         \label{fig:cliqueForestOfG-U}
\end{figure}

\noindent
Using essentially the same argument, a similar lemma for pendant paths can be obtained:

\begin{lemma}\label{lem:pendant-path-removal}
	Let $\mathcal{T} = (\mathcal{C}, \mathcal{E})$ be the clique forest of a chordal graph $G=(V, E)$.
	Let $\mathcal{P} = C_1, C_2, \dots, C_k$ be a pendant path in $\mathcal{T}$.
	Then $\mathcal{T} - \mathcal{P}$ is the clique forest of $G[V \setminus U]$, where 
	$U = \{u \in V \, : \, \mathcal{T}(u) \mbox{ is a subpath of } \mathcal{P} \}$.
\end{lemma}

\begin{lemma}\label{lem:tree-decomp}
 For every $i$, $\mathcal{T}_i$ is the clique forest of $G[U_i]$.
\end{lemma}
\begin{proof}
We prove this statement by induction on $i$. The base case $i=1$ follows by definition 
of $\mathcal{T}_1$ (recall that $\mathcal{T}_1 = \mathcal{T}$ and $U_1 = V$). For the induction step,
we apply Lemma~\ref{lem:internal-path-removal} for every internal path in $\mathcal{L}_i$ of diameter at least $3k$,
and Lemma~\ref{lem:pendant-path-removal} for every pendant path in $\mathcal{L}_i$.
\end{proof}

Using the fact that $\mathcal{T}_{i+1}$ is obtained from $\mathcal{T}_i$ by removing all pendant and some internal paths,
we show now that the first phase of our algorithm requires at most $\lceil \log n \rceil$ iterations. This is a consequence of
the following pruning lemma:

\begin{lemma}[Pruning Lemma]\label{lem:forestPruning}
	Let $T_1$ be an $n$-node forest, and for every $i \geq 2$, let $T_i$ be a forest obtained from $T_{i-1}$ by removing
	all its maximal pendant paths and some of its maximal internal paths. Then $T_i$ has less than $n/2^i$ nodes
	of degree at least $3$. In particular, $T_{\lceil \log n \rceil +1}$ has no nodes.
\end{lemma}
\begin{proof}
	The maximality of the removed paths implies that the nodes of degree at most $1$ in $T_{i}$ have degree at least 3 in
	$T_{i-1}$.
	Therefore denoting by $L_i$ and $A_i$ the set of nodes of degree at most $1$ and the set of nodes of degree at least $3$
	in $T_{i}$, respectively, we have $L_i \subseteq A_{i-1}$. Furthermore, it is obvious that $A_i \subseteq A_{i-1}$, and hence
	\begin{equation}\label{eq:AiB}	
		|L_{i}| + |A_{i}| \leq |A_{i-1}|.
	\end{equation}
	Since the number of nodes of degree at least $3$ in a forest is less than the number of nodes of degree at most 1,
	we have that $|A_1| < n/2$.
	For the same reason $|A_i| < |L_i|$, which together with inequality~(\ref{eq:AiB}) implies
	$|A_i| < |A_{i-1}|/2$ for every $i \geq 2$. This implies the lemma.
\end{proof}


\begin{corollary}\label{lem:pruning}
	Phase 1 in Algorithm \ref{alg:coloringChordal} requires at most $\lceil \log n \rceil$ iterations, i.e., 
	$\bigcup_{1 \le i \le \lceil \log n \rceil} V_i = V$.
\end{corollary}

Next, we address the color correction phase. In each iteration of this phase we consider every path 
$\mathcal{P} \in \mathcal{L}_i$ independently. The subgraph induced by the set of nodes $W \subseteq V_i$ whose corresponding 
trees are subpaths of $\mathcal{P}$ is legally colored in the coloring phase. This coloring may be inconsistent with the 
coloring of subgraph $G[U_{i+1}]$. We first prove in Lemma~\ref{lem:conflict-vertices} that the set $W' \subseteq \bigcup_{s > i} V_s = U_{i+1}$ 
of neighbors of $W$ in $G[U_{i+1}]$ (i.e., the nodes in $U_{i+1}$ that could potentially cause conflicts) 
is the union of at most two cliques, which are included in the end vertices of the clique forest of interval graph $G[W \cup W']$.
This lemma makes use of Lemma~\ref{lem:helpful-1}, which shows that the set of nodes whose corresponding subtrees in $\mathcal{T}_i$
are contained in an arbitrary path in $\mathcal{T}_i$ form an interval graph.
\begin{lemma}\label{lem:helpful-1}
	Let $\mathcal{T} = (\mathcal{C}, \mathcal{E})$ be the clique forest of a chordal graph $G=(V, E)$. 
	Let $\mathcal{P} = C_1, C_2, \dots, C_k$ be a path in $\mathcal{T}$, and let $V_{\mathcal{P}} = \bigcup_{i=1}^k C_i$ 
	be the set of nodes whose corresponding subtrees intersect with $\mathcal{P}$. 
	Then $\mathcal{P}$ is the clique forest of $G[V_{\mathcal{P}}]$, and $G[V_{\mathcal{P}}]$ is an interval graph.
\end{lemma}
\begin{proof}
	By Lemma \ref{cl:unique-subtree}, the maximum weight spanning forest in  $\mathcal{W}_G[\{ C_1, \ldots, C_k \}]$ coincides 
	with path $\mathcal{P} = \mathcal{T}[\{ C_1, \ldots, C_k \}]$. 
	Since $ C_1, \dots, C_k$ are the maximal cliques of $G[V_{\mathcal{P}}]$, we conclude that $\mathcal{P}$ is 
	the clique forest of $G[V_{\mathcal{P}}]$,
	and by Theorem \ref{th:interval-char} $G[V_{\mathcal{P}}]$ is an interval graph.  
\end{proof}

\begin{lemma}\label{lem:conflict-vertices}
	Let $\mathcal{T} = (\mathcal{C}, \mathcal{E})$ be the clique forest of a chordal graph $G=(V, E)$.
	Let $\mathcal{P} = C_1, C_2, \dots, C_k$ be a path in $\mathcal{T}$.
	Let $W \subseteq V$ be the subset of nodes whose corresponding subtrees are subpaths of $\mathcal{P}$, and let
	$W' \subseteq V \setminus W$ be the set of nodes outside $W$ that have neighbors in $W$. 
	Then $G[W \cup W']$ is an interval graph, and
	\begin{enumerate}
		\item if $\mathcal{P}$ is an internal path connected to $\mathcal{T}$ by edges $C_sC_1$ and $C_kC_e$, then
		 $W' \subseteq C_s \cup C_e$, and the cliques $W' \cap C_s$ and $W' \cap C_e$ are contained in the end vertices 
		of the clique forest of $G[W \cup W']$;

		\item if $\mathcal{P}$ is a pendant path connected to $\mathcal{T}$ by edge $C_kC_e$, then
		$W' \subseteq C_e$, and $W' \cap C_e$ is contained in one of the end vertices of the clique forest of $G[W \cup W']$.
	\end{enumerate}
\end{lemma}
\begin{proof}
	First, notice that $W \cup W'$ is a subset of nodes $v$ whose subtrees $\mathcal{T}(v)$ intersect with $\mathcal{P}$.
	Hence, it follows from Lemma \ref{lem:helpful-1} that $G[W \cup W']$ is an interval graph.
	
	Let now $\mathcal{P}$ be an internal path.
	It follows from the definition of clique forest that $W'$ is a subset of $C_s \cup C_e$.
	Since $W \cup C_s \cup C_e = \bigcup_{i=1}^{k} C_i \cup C_s \cup C_e$, and
	$C_s, C_1, \ldots, C_k, C_e$ is a path in $\mathcal{T}$, by Lemma \ref{lem:helpful-1}
	graph $G[W \cup C_s \cup C_e]$ is an interval graph with the clique forest being path 
	$C_s, C_1, C_2, \dots, C_k, C_e$. As $W' \subseteq C_s \cup C_e$ we conclude that $G[W \cup W']$ is 
	an interval graph as well. Finally, it is not hard to see that every node of $W'$ is necessarily contained in one of the
	end vertices of the clique forest of $G[W \cup W']$. 
	The case of $\mathcal{P}$ being a pendant path is proved similarly. 
\end{proof}

In order to resolve these coloring conflicts we carry out a recoloring process on interval graph $G[W \cup W']$ with
fixed colorings of its 'boundary' cliques. To this end, we reuse a result by Halld\'{o}rsson and Konrad \cite{hk17}:\footnote{The original lemma proved by Halld\'{o}rsson 
and Konrad is formulated slightly differently. The version needed here can be proved following exactly the proof given in \cite{hk17}.}

\begin{lemma}[Halld\'{o}rsson and Konrad \cite{hk17}]\label{lem:col-completion}
Let $G=(V,E)$ be an interval graph with its clique forest $\mathcal{T} = (\mathcal{C}, \mathcal{E})$ being a path $\mathcal{P} = C_1, C_2, \ldots, C_k$
such that $\textup{dist}_G(u, v) \ge r$ for every pair of nodes $u \in C_1, v \in C_k$, for an integer $r \ge 5$.
Suppose that cliques $C_1$ and $C_k$ are legally colored using at most $c$ colors.
Then the coloring of $G[C_1 \cup C_k]$ can be extended to a legal coloring of $G$ with at most 
$\max \{ \lfloor (1+\frac{1}{r-3}) \chi(G) \rfloor + 1, c \}$ colors.
\end{lemma}

Equipped with Lemma~\ref{lem:col-completion}, we now prove correctness of the color correction phase.
 
\begin{lemma}[Recoloring Lemma]\label{lem:recoloring}
 Consider the color correction phase (Step 3) of Algorithm \ref{alg:coloringChordal}. 
  Let $\mathcal{P} \in \mathcal{L}_i$ be a path and let $W \subseteq V_i$ 
  be the subset of nodes whose corresponding subtrees are included in $\mathcal{P}$. 
  Further, let $W' \subseteq \bigcup_{s > i} V_s = U_{i+1}$ be the nodes in $U_{i+1}$
  that have neighbors in $W$.
  Suppose that $W'$ is colored using colors from the set $[ \lfloor (1+1/k)\chi(G) + 1 \rfloor ]$.
  Then, we can recolor those nodes of $W$ that are at distance at most $k+4$ from $W'$ in $G$ with colors from the set 
$[ \lfloor (1+1/k)\chi(G) + 1 \rfloor ]$ so that $G[W \cup W']$ is legally colored.
\end{lemma}
\begin{proof}
By Lemma~\ref{lem:conflict-vertices}, $G[W \cup W']$ is an interval graph and its clique forest is a path. 
Let $\mathcal{P}' = C_1, C_2, \dots, C_r$ denote this path. 
Lemma~\ref{lem:conflict-vertices} also states that $W' \subseteq C_1 \cup C_r$. 
 
 Let $i$ be the minimum index such that $\dist(u,v) \ge k+3$, for every
 $u \in W' \cap C_1$ and $v \in C_i$. Then, by Lemma~\ref{lem:col-completion}, the nodes of the cliques $C_2, \dots, C_{i-1}$ 
 can be recolored using at most  $\lfloor (1+1/k) \chi(G) \rfloor + 1$ colors to resolve the coloring conflicts
 between $W' \cap C_1$ and $W$. Similarly, let $j$ be the maximum index such that $\dist(u,v) \ge k+3$, for every 
 $u \in W' \cap C_r$ and $v \in C_j$. Then, again by Lemma~\ref{lem:col-completion}, the nodes of the cliques 
$C_{j+1}, \dots, C_{r-1}$ can be recolored using at most  $\lfloor (1+1/k) \chi(G) \rfloor + 1$ colors to resolve the coloring conflicts
 between $W' \cap C_r$ and $W$. 
\end{proof}

\begin{theorem} \label{thm:sequential}
 For every $\epsilon > \frac{2}{\chi(G)}$, Algorithm~\ref{alg:coloringChordal} is a $(1+\epsilon)$-approximation \textsc{MVC} algorithm on chordal graphs.
\end{theorem}
\begin{proof}
 First, we show by induction that the algorithm uses at most $(1+1/k)\chi(G) + 1$ colors. 
This is clearly true for $G_{\lceil \log n \rceil}$. 
The induction step follows from Lemma~\ref{lem:recoloring}. Now, using the assumption $\epsilon > \frac{2}{\chi(G)}$, we obtain:
$(1+1/k)\chi(G) + 1 \le (1+\epsilon/2)\chi(G) + \epsilon \chi(G)/2 = (1+\epsilon) \chi(G)$,
which proves the approximation factor of the algorithm. 
\end{proof}

 \section{Minimum Vertex Coloring: Distributed Implementation}
 \label{sec:coloring-distributed}
 We now give an implementation of Algorithm~\ref{alg:coloringChordal} in the {\sf LOCAL} model.
 
 \subsection{Algorithm}
 The global behavior of our distributed algorithm, Algorithm~\ref{alg:distColoringChordal}, is identical to that of our centralized Algorithm~\ref{alg:coloringChordal}. 
 The main challenge lies in the coordination of the network nodes. One particular difficulty stems from the fact
 that network nodes are not aware of $n$, the total number of nodes, and thus do not know when the 
 $\lceil \log n \rceil$ iterations of the pruning phase have completed. For this reason, nodes execute the three phases of
 Algorithm~\ref{alg:coloringChordal} asynchronously.
 
 We will first present the pseudocode of our distributed algorithm, which is executed independently on every node $v$.
 Then we will describe each of the three phases in detail.

\begin{algfloat} 
 \normalsize
 \noindent \fbox{\parbox{\textwidth - 7pt}{
 
 \textbf{Input:} 
a parameter $\epsilon$.

\vskip1ex

Let $k = \lceil 2/\epsilon \rceil$.

\vspace{-0.2cm}

 \begin{enumerate}
\item \textbf{Pruning Phase.} 
$(l_v, parent_v, children_v) \gets \textsc{pruneTree}()$. 


\item \textbf{Coloring Phase.} 
Run $\textsc{ColIntGraph}(\frac{1}{k})$ on layer $l_v$ and store color in $c_v$.


\item \textbf{Color Correction Phase.} 
\begin{enumerate}[topsep=0pt,leftmargin=0cm]
 \item[] \textbf{if} $parent_v \neq \bot$ \textbf{then} 
 
\begin{enumerate}[topsep=0pt,leftmargin=0.5cm]
	\item[] Wait until message $\textsc{SetColor}(c)$ received from $parent_v$.
	\item[] Set $c_v \gets c$.
 \end{enumerate}

 \textbf{end if}
\item[] $\textsc{CorrectChildren}(children_v,k)$.
\end{enumerate}
 \end{enumerate}
 \vspace{-0.3cm}
}}
\begin{center}
\algo{alg:distColoringChordal}
 \textbf{Algorithm \ref*{alg:distColoringChordal}.} A distributed $(1+\epsilon)$-approximation algorithm, code for node $v$.
\end{center}
\vspace{-0.5cm}
\end{algfloat}


\vskip1ex
\mypara{The Pruning Phase}
In the pruning phase, the subroutine $\textsc{PruneTree}$ is invoked and returns parameters $l_v$, $parent_v$, and $children_v$,
where $l_v$ is the layer of node $v$, and $parent_v$ and $children_v$ are variables necessary for the coordination of the
color correction phase and are defined and explained further below. The pseudocode of \textsc{PruneTree} is given in 
Algorithm~\ref{alg:pruneTree}. 

\begin{algfloat}
 \normalsize
 \noindent \fbox{\parbox{\textwidth - 7pt}{
 
 \textbf{Initialization: } \\
 Let $i = 1$, $l_v = -1, children_v = \{ \}$, and $parent_v = \bot$.
 
 \vspace{0.1cm}
 \textbf{while} $l_v = -1$ \textbf{do}:
 \begin{enumerate}[topsep=4pt,leftmargin=0.8cm]
  \item Collect $\Gamma^{10k}_G(v)$ together with variables $l_u$ and $ID_u$, for every $u \in \Gamma^{10k}_G(v)$.
  \item Compute local view of the clique forest $\mathcal{T}_i = (\mathcal{C}_i, \mathcal{E}_i)$ of the subgraph of $G$
	induced by the nodes $u \in \Gamma^{10k}_G(v)$ with $l_u = -1$.
  
  \item \textbf{if} $\mathcal{T}_i(v)$ is a subpath of a pendant path in $\mathcal{T}_i$, or
			  $\mathcal{T}_i(v)$ is a subpath of a binary path in $\mathcal{T}_i$ of diameter at least $3k$ \textbf{then}
	\begin{enumerate}[topsep=0pt,leftmargin=0.5cm]
   		\item[] $l_v = i$.
		\item[] $parent_v = $ parent of $v$.
	\end{enumerate}
   \textbf{else}
	\begin{enumerate}[topsep=0pt,leftmargin=0.5cm]
  	 	\item[] Add children in layer $i$ (if there are any) to $children_v[i]$.
	\end{enumerate}
   \textbf{end if}
  \item $i = i + 1$.
 \end{enumerate}
 \vspace{0.1cm}
\textbf{return} $(l_v, parent_v, children_v)$.
}}
\begin{center}
\algo{alg:pruneTree}
 \textbf{Algorithm \ref*{alg:pruneTree}.} \textsc{PruneTree}(), code for node $v$.
\end{center}
\vspace{-0.5cm}
\end{algfloat}

 In each iteration of the while loop of \textsc{PruneTree}, one layer is removed from the clique forest of the input graph. 
 To describe the global behavior of the algorithm, we will reuse the naming conventions already used in Algorithm~\ref{alg:coloringChordal}. 
 Let $U_1 = V$, and let $V_i \subseteq U_i$ be the set of nodes removed in iteration $i$, i.e., assigned layer index $i$. 
Let also $U_{i+1} = U_i \setminus V_i$, and 
let $\mathcal{L}_i$ be the set of maximal paths removed from the clique forest $\mathcal{T}_i$ of $G[U_i]$. 
For convenience we will sometimes denote $G[U_i]$ by $G_i$.
 In each iteration $i$, first, each node $v$ collects its distance-$10k$ neighborhood. Then, $v$ computes its local view of 
 the clique forest $\mathcal{T}_i$ of the graph induced by the nodes that have not yet been removed from the graph, i.e., 
 of $G[U_i]$ (as in Section~\ref{sec:local-views}).
 Next, nodes $v$ are removed from $G[U_i]$ and added to the current layer $V_i$ if its corresponding subtrees $\mathcal{T}_i(v)$ 
 are entirely contained in either a pendant path or a binary path of large enough diameter. This step is identical to 
Algorithm~\ref{alg:coloringChordal}, and 
 the exact same partitioning is computed. Nodes $v$ that are removed in the current iteration store their {\em parent} 
 in variable $parent_v$, and nodes that remain in the graph potentially store some of the removed nodes as their children in 
 variable $children_v$. The notions of parent and child are defined as follows: 
 
 \begin{definition}[Parent, Child]
Let $v \in V_i$ and
let $\mathcal{P}$ be the maximal 
binary
path in $\mathcal{T}_i$ that contains $\mathcal{T}_i(v)$. If $\mathcal{P}$ is a component of $\mathcal{T}_i$
then we define $parent_v := \bot$. Otherwise, let $C$ be 
the vertex outside $\mathcal{P}$ in $\mathcal{T}_i$ such that $C$ is adjacent to an end vertex of $\mathcal{P}$ and 
$\dist_G(v, C)$ is minimal. 
Let $c$ be the node with maximum ID in $C$. Then the \textit{parent} of $v$ is defined to be $c$, if $dist_G(v, C) \le k + 3$,
and $\bot$ otherwise.
If $c$ is the parent of $v$, then we say that $v$ is the \textit{child} of $c$.
\end{definition}
 The parent of node $v$ is responsible for recoloring $v$ in the color correction phase. Notice that a node $v$ does not have a parent
 if the closest maximal clique outside $v$'s path $\mathcal{P}$ is at least at distance $k+4$ from $v$. In this case, the color that $v$ 
 will receive in the coloring phase is final and no color correction is needed for $v$. Recall that in the color correction phase of Algorithm~\ref{alg:coloringChordal},
 we only need to recolor nodes that are at distance at most $k+3$ from the cliques that contain nodes with color conflicts.
 Finally, the subroutine returns the node's level $l_v$, its parent $parent_v$, and its children $children_v$.

\mypara{The Coloring Phase}
Notice that all nodes of layer $i$ return from the subroutine \textsc{PruneTree} in the same round. They can hence invoke the 
coloring phase simultaneously. They run the algorithm \textsc{ColIntGraph} of Halld\'orsson and Konrad \cite{hk17} 
and compute a coloring on $G[V_i]$ that uses at most $\lfloor (1 + \frac{1}{k}) \chi(G[V_i] + 1 \rfloor$ colors. 
This algorithm runs in $\Order(k \log^* n)$ rounds.
 
 While some nodes execute the coloring phase, others still proceed with the execution of $\textsc{PruneTree}$. These nodes
 repeatedly collect their distance-$10k$ neighborhood. This requires all other network nodes, in particular, nodes that already
 completed the first phase, to continuously forward messages. This can be taken care of in the background, and we will not
 address this type of implementation detail any further. 
 
\mypara{The Color Correction Phase} In the color correction phase, nodes with assigned parents (i.e., nodes $v$ with $parent_v \neq \bot$)
 first wait until they received their final color from their parents. Only then they proceed and correct the colors of their children. 
 To this end, each such node $v$ runs the subroutine $\textsc{CorrectChildren}$, which processes $children_v$ layer by layer,
 starting with layer $l_v - 1$ down to $1$. If $v$ has children in layer $V_i$, then it waits until all nodes adjacent to $children_v[i]$
 which are contained in layers $> i$ have received their final colors. This can be done by repeatedly collecting its local 
 distance-$(k+5)$ neighborhood and monitoring whether the colors of all nodes in $\Gamma_{G_i}(children_v[i])$ are final. 
 Then, $v$ locally computes the color correction for $children_v[i]$ and notifies them about their new colors. 
 
 \begin{algfloat}
 \normalsize
 \noindent \fbox{\parbox{\textwidth - 7pt}{
 \textbf{for} $l \gets l_v - 1, l_v - 2, \dots, 1$ \textbf{do:}
 \begin{enumerate}[topsep=0pt,leftmargin=0.5cm]
 	\item[] \textbf{if} $children_v[l] \neq \{ \}$ \textbf{then}
  	\begin{enumerate}[topsep=0pt,leftmargin=1cm]
 		 \item Wait until all neighbors of $children_v[l]$ in $G[U_l]$ have received their final color.
  		\item Compute color correction as in Lemma~\ref{lem:recoloring}.
 		 \item For each $u \in children_v[l]$, send message $\textsc{SetColor}(c)$ to $u$, where $c$ is $u$'s new color.
 	\end{enumerate}
 \end{enumerate}
}}
\begin{center}
\algo{alg:correctChildren}
 \textbf{Algorithm \ref*{alg:correctChildren}.} $\textsc{CorrectChildren}(children_v, k)$, code for node $v$.
\end{center}
\vspace{-0.5cm}
\end{algfloat}

\subsection{Analysis}
To ensure correctness of our algorithm, we need to show that the parent of a node $v \in V_i$ is contained in a layer $j > i$ (Corollary~\ref{cor:parent-layer}).
This is proved via the following lemma:

 
\begin{lemma} \label{lem:neighbors-layers}
Let $\mathcal{P} \in \mathcal{L}_i$, and let $W \subseteq V_i$ be the set of nodes whose corresponding subtrees are included in 
$\mathcal{P}$. Then every node $u \in \Gamma_{G_i}(W)$ is contained in a layer $V_j$ with $j > i$.
\end{lemma}
\begin{proof}
	Notice that if a node $v$ is assigned layer number $i$, then $\mathcal{T}_i(v)$ is a subpath of a binary path in $\mathcal{T}_i$.
	Now let $u \in \Gamma_{G_i}(W)$ be a neighbor of some node $v \in W$. Since $u \notin W$, 
	$\mathcal{T}_i(u)$ is not a subpath of $\mathcal{P}$. On the other hand, the fact that $u$ is a neighbor of $v$ implies that
	$u$ belongs to a vertex $C$ of $\mathcal{T}_i$, which is adjacent to one of the end vertices of $\mathcal{P}$.
	As $\mathcal{P}$ is a maximal binary path, $C$ has degree at least 3 in $\mathcal{T}_i$. Therefore $\mathcal{T}_i(u)$
	is not a subpath of a binary path in $\mathcal{T}_i$, and hence $u$ is not assigned a layer number in the $i$-th iteration, which
	implies the lemma. 
\end{proof} 

\noindent It follows from the definition of parent that the parent of a node $v \in V_i$ belongs to $\Gamma_{G_i}(W)$. Hence:
\begin{corollary} \label{cor:parent-layer}
	The parent of a node $v \in V_i$ is contained in a layer $V_j$ with $j > i$.
\end{corollary}

We next demonstrate that Algorithm~\ref{alg:distColoringChordal} mimics the behavior of our centralized algorithm and uses $\Order(\frac{1}{\epsilon} \log n)$ rounds.
This establishes our main result, which is stated in Theorem~\ref{thm:dist-alg}.
\begin{lemma} \label{lem:dist-centr-identical}
	The global behavior of Algorithm~\ref{alg:distColoringChordal} is identical to the behavior of Algorithm~\ref{alg:coloringChordal}. 
	Furthermore, Algorithm~\ref{alg:distColoringChordal} runs in $\Order(\frac{1}{\epsilon} \log n)$ rounds.
\end{lemma}
\begin{proof}
We first argue that Algorithm~\ref{alg:distColoringChordal} computes the same node set partition 
as Algorithm~\ref{alg:coloringChordal}. Indeed,
a node assigns itself a layer number if the condition in Step~3 of the while loop in algorithm \textsc{PruneTree} is fulfilled. 
This condition is equivalent
to the removal condition in Algorithm~\ref{alg:coloringChordal}. Furthermore, note that 
the node's parent and the node's children are at distances at most $k+4$ from a node. Since nodes possess knowledge about their 
distance-$10k$ neighborhoods, computing and storing the parent and children can be done locally. Concerning the runtime, every node 
of layer $i$ exits $\textsc{PruneTree}$ after $i$ iterations of the while loop, which each requires $\Order(k)$ rounds.
Since by Lemma~\ref{lem:pruning} the number of layers is bounded by $\lceil \log n \rceil$, after $\Order(k \log n)$ rounds, every node has returned from $\textsc{PruneTree}$.
    
Next, nodes of the same layer exit the first phase simultaneously and then execute the coloring algorithm \textsc{ColIntGraph} of 
Halld\'{o}rsson and Konrad. This assigns each node $v$ a color $c_v$ similar to Algorithm~\ref{alg:coloringChordal}. Algorithm \textsc{ColIntGraph} 
runs in $\Order(k \log^* n)$ rounds. Hence, overall after $\Order(k \log n)$ rounds, every node has completed the coloring step.
  
Let $t_{max}$ be the number of rounds when the last node has completed the coloring phase. Note that $t_{max} \leq r k \log n$
for some constant $r$. We prove by 
induction that after $t_{max} + i \Order(k)$ rounds, all nodes of layers $\lceil \log n \rceil - i, \dots, \lceil \log n \rceil$ 
have received their final colors. First, observe that nodes of layer $\lceil \log n \rceil$ 
have received their final colors already in the coloring phase. Suppose now that all nodes of layers 
$\lceil \log n \rceil - i, \dots, \lceil \log n \rceil$ have received their final colors. Let $v \in V_{\lceil \log n \rceil - i - 1}$ 
be a node that has a parent $u$ (otherwise, $v$ has received its final color already in the coloring phase),
let $\mathcal{P} \in \mathcal{L}_{\lceil \log n \rceil - i - 1}$ be the path that contains $\mathcal{T}_{\lceil \log n \rceil - i - 1}(v)$, 
and let $W$ be the set of nodes whose subtrees are contained in $\mathcal{P}$. 
By Lemma~\ref{lem:neighbors-layers}, all neighbors of $W$ outside $W$ are contained in a layer with index greater than 
$\lceil \log n \rceil - i - 1$, and hence, by the induction hypothesis, all of them
  have received their final colors after $t_{max} + i \Order(k)$ rounds. The same applies to $v$'s parent $u$, by Corollary~\ref{cor:parent-layer}. Node $u$ hence 
  begins the execution of \textsc{CorrectChildren} no later than in iteration $t_{max} + i \Order(k) + 1$. 
Since it takes at most $\Order(k)$ rounds to collect the local neighborhood and inform nodes about their new colors, 
  after $t_{max} + (i+1) \Order(k)$ rounds, the nodes of layer $V_{\lceil \log n \rceil - i - 1}$ have received their final color. 
This completes the induction.
  
  Hence, the runtime of the algorithm is $\Order(k \log n) = \Order(\frac{1}{\epsilon} \log n)$, which completes the proof.
 \end{proof}

 

 \begin{theorem} \label{thm:dist-alg} 
 For every $\epsilon \ge \frac{2}{\chi(G)}$, there is a deterministic $(1+\epsilon)$-approximation algorithm for 
 \textsc{MVC} on chordal graphs that runs in $\Order(\frac{1}{\epsilon} \log n)$ rounds in the {\sf LOCAL} model.
\end{theorem}

\section{Maximum Independent Set on Interval Graphs}\label{sec:MISinterval}
As an intermediate step towards a distributed $(1 + \epsilon)$-approximation algorithm for \textsc{MIS} 
on chordal graphs, we provide in this section a distributed $(1 + \epsilon)$-approximation algorithm 
for \textsc{MIS} on interval graphs.


\subsection{Centralized Algorithm}
\label{sec:mis_int} \label{sec:unit_mis}
In this section we give a deterministic distributed $(1+\epsilon)$-approximation algorithm for \textsc{MIS} on interval graphs
that runs in $\Order(\frac{1}{\epsilon} \log^* n)$ rounds in the {\sf LOCAL} model. Our algorithm is stated as a centralized
algorithm in Algorithm~\ref{alg:misInterval}. Its implementation in the {\sf LOCAL} model is straightforward and will
be discussed further below. 

Let $H=(V, E)$ be an interval graph. Let $V_P \subseteq V$ be the set of nodes $v$ such that there exists a node
$u \in V$ with $\Gamma_H[v] \supsetneq \Gamma_H[u]$. Notice that nodes in $V_P$ can be ignored when computing a maximum 
independent set, since if a maximum independent set contains a node of $V_P$, then there always exists a node outside $V_P$ 
that could be included instead. We therefore first remove from $H$ all nodes in $V_P$. 
It is easy to see that the resulting graph $H'$ is a proper interval graph, and therefore is a unit interval graph
\cite{roberts1969indifference}. 
Next, we find a distance-$k$ maximal independent set $I_1$ in $H'$, 
and compute maximum independent sets between any two consecutive intervals in $I_1$. As we will see from 
the analysis, the union of these maximum independent sets gives us a desired approximate maximum independent set in $H$.

To construct a distance-$k$ maximal independent set in a unit interval graph we employ the fact
that for any natural number $k$, the $k$-th power of a unit interval graph is a unit interval graph 
\cite{raychaudhuri1987powers}. We use the $\Order(\log^* n)$ rounds distributed algorithm of Schneider and Wattenhofer 
\cite{sw10}, which can be used to compute a maximal independent set in a unit interval graph. We will denote the latter 
algorithm by \textsc{MISUnitInterval}.

\begin{algfloat}
 \normalsize
 \noindent \fbox{\parbox{\textwidth - 7pt}{
 
 \textbf{Input: } An interval graph $H=(V, E)$ on $n$ nodes; a parameter $\epsilon \in (0,1)$.
 \begin{enumerate}[leftmargin=0.2cm]
 \item[] $k = \lceil 2.5/\epsilon + 0.5 \rceil$.
 \item[] $I = \emptyset$.
 \item[] Remove from $H$ all nodes $v$ such that there exists a node $u \in V$ with $\Gamma_H[v] \supsetneq \Gamma_H[u]$.
 \item[] \textbf{for} every maximal connected subgraph $G$ of the resulting graph \textbf{do}
	\begin{enumerate}
 \item[] \textbf{if} $diam(G) \le 10k$ \textbf{then} 
	\begin{enumerate}[leftmargin=1cm]
		\item[1.] Compute a maximum independent set $I^*$ of $G$.
		\item[2.] $I = I \cup I^*$.
	\end{enumerate}
 \item[] \textbf{else}
		\begin{enumerate}[leftmargin=1cm]
  \item[1.] Compute a distance-$k$ maximal independent set $I_1$ in $G$ (in the distributed setting simulate \textsc{MISUnitInterval}
  on $G^k$ in $\Order(k \log^* n)$ rounds).
  \item[2.] Let $P = \{(u,v) \in I_1 \times I_1 \, : \, u \neq v \mbox{ and } dist_G(u,v) \le 2k - 1 \}$. 
  \item[3.] For every $(u,v) \in P$, let $V_{u,v} \subseteq V \setminus (\Gamma_G[u] \cup \Gamma_G[v])$ be the set of nodes $w$ with $\max \{dist_G(u,w), dist_G(v,w) \} \le dist_G(u,v)$ ($V_{u,v}$ is located between $u$ and $v$).
  \item[4.] Nodes $u,v$ compute a maximum independent set $I_{u,v}$ in $G[V_{u,v}]$.
  \item[5.] Furthermore, let $v_{\ell}, v_r \in I_1$ be the left-most and right-most intervals of $I_1$, respectively. Interval $v_{\ell}$ ($v_r$) 
  computes a maximum independent set $I_{\ell} $ (resp. $I_r$) using intervals located to the left (resp. right) of $v_{\ell} $ (resp. $v_r$).
  \item[6.] $I = I \cup I_1 \cup I_{\ell} \cup I_r \cup \bigcup_{(u,v) \in P} I_{u,v}$.
		\end{enumerate}

	\end{enumerate}
  \item[] \textbf{return} $I$.
  \end{enumerate}
}}
\begin{center}
\algo{alg:misInterval}
 \textbf{Algorithm \ref*{alg:misInterval}.} A deterministic centralized $(1+\epsilon)$-approximation algorithm 
 for the maximum independent set problem in interval graphs.
\end{center}
\vspace{-0.5cm}
\end{algfloat}

\begin{theorem}
 For every $\epsilon \in (0,1)$ Algorithm~\ref{alg:misInterval} gives a $(1+\epsilon)$-approximation algorithm 
 for \textsc{MIS} on interval graphs.
\end{theorem}
\begin{proof}
Let  $k = \left\lceil 2.5/\epsilon + 0.5 \right\rceil$.
To prove the approximation factor it is enough to prove it for a maximal connected subgraph $G$ of $H$. 
If $diam(G) \le 10k$, then the algorithm computes a maximum independent set. Therefore we assume that $diam(G) > 10k$.
Let $I^*$ denote a maximum independent set in $G$, and
let $I^*_1 \subseteq I^*$ be those intervals that intersect with some interval in $I_1$. 
Since $G$ is a unit interval graph (in particular, claw-free), we have $|I^*_1| \le 2|I_1|$. 
 For every $(u,v) \in P$, 
 let $I^*_{u,v} \subseteq I^*$ be those intervals that lie between $u$ and $v$ and are not adjacent to 
 $u$ and $v$. Then, $I^*_{u,v} \subseteq V_{u,v}$ holds, and, since $I_{u,v}$ is a maximum independent set in $V_{u,v}$, we have
 $|I^*_{u,v}| \le |I_{u,v}|$. Furthermore, let $I^*_l, I^*_r \subseteq I^*$ be the remaining intervals outside 
$I^*_1 \cup \bigcup_{(u,v) \in P} I^*_{u,v}$ on the left of $v_l$ and on the right of $v_r$, respectively. 
 Then, $|I^*_l| \le |I_l|$ and $|I^*_r| \le |I_r|$. 

 Next, let $\beta = \sum_{(u,v) \in P} |I_{u,v}|$. Observe that for every $(u,v) \in P$, we have $|I_{u,v}| \ge (k-3) / 2$.  
 We thus obtain $\beta \geq |P| (k-3) / 2$. Furthermore, notice that $|I_1| = |P| + 1$.
 Finally, since $diam(G) > 10k$, we have $|I_1| \ge 5$ and hence $|P| \ge 4$. 
 This in turn implies that $\beta \ge 2 (k-3)$. 

Using the inequalities argued above, we can bound the approximation factor as follows:

 \begin{eqnarray*}
  \frac{|I^*|}{|I|} & \le & \frac{|I^*_1| + |I^*_l| + |I^*_r| + \sum_{(u,v) \in P} |I^*_{u,v}| }{|I_1| + |I_l| + |I_r| + \sum_{(u,v) \in P} |I_{u,v}|} \le 
  \frac{2 |I_1| + |I_l| + |I_r| + \sum_{(u,v) \in P} |I_{u,v}| }{|I_1| + |I_l| + |I_r| + \sum_{(u,v) \in P} |I_{u,v}|} \\
  & \le & \frac{2 (|P| + 1) + \beta}{(|P| + 1) + \beta} \le \frac{2(\frac{2 \beta}{k-3} + 1) + \beta }{(\frac{2 \beta}{k-3} + 1)+ \beta} = 
  \frac{\beta(1 + \frac{4}{k-3}) + 2}{\beta(1+\frac{2}{k-3}) + 1} \le
  \frac{\beta(1 + \frac{4}{k-3}) + 2 \frac{\beta}{2(k-3)}}{\beta(1+\frac{2}{k-3}) + \frac{\beta}{2(k-3)}} \\
  & = & \frac{1+\frac{5}{k-3}}{1+\frac{2.5}{k-3}} = 1 + \frac{2.5}{k - 0.5} .
 \end{eqnarray*}

\noindent
Since $k \ge 2.5/\epsilon + 0.5$, we obtain a $(1+\epsilon)$-approximation. 
\end{proof}

\subsection{Distributed Implementation}\label{sec:intDistImpl}

Observe that nodes $v \in V$ can check locally both whether there exists a node $u \in V$ 
with $\Gamma_H[v] \supsetneq \Gamma_H[u]$ and whether the diameter of the input graph is at least $10k$. Simulating 
algorithms on the $k$-th power of the input graph incurs an additional factor of $k$ in the round complexity. All steps except
the computation of the distance-$k$ independent set require $\Order(k)$ rounds, while the latter requires $\Order(k \log^* n)$ 
rounds. This gives the following theorem:

\begin{theorem}\label{thm:mis-interval-graphs}
	For every $\epsilon > 0$, there is a deterministic $(1+\epsilon)$-approximation algorithm for \textsc{MIS} on interval graphs that 
	operates in $\Order(\frac{1}{\epsilon} \log^* n)$ rounds in the {\sf LOCAL} model.
\end{theorem}


\section{Maximum Independent Set on Chordal Graphs}\label{sec:mis-chordal}

In this section we present a distributed $(1 + \epsilon)$-approximation algorithm for \textsc{MIS}
on chordal graphs that runs in $\Order(\frac{1}{\epsilon}\log(\frac{1}{\epsilon})\log^* n)$ rounds 
in the {\sf LOCAL} model. Similarly to Minimum Vertex Coloring, we first provide and analyze a centralized algorithm (Sections~\ref{sec:chCentralizedImpl} and \ref{sec:chAlgAnalysis}), and then briefly discuss a distributed implementation of the algorithm (Section~\ref{sec:chDistImpl}).

\subsection{Centralized Algorithm}\label{sec:chCentralizedImpl} 

Similar to our coloring algorithm, we iteratively peel off binary paths from the clique forest of an input graph.
However, instead of peeling off all layers in $\Order(\log n)$ iterations, we stop after $k$ iterations, for 
some $k = \Theta(\log \frac{1}{\epsilon})$. We will show that the set of removed nodes in these $k$ iterations 
already contains an almost optimal independent set. For convenience, we denote by $G_i$ graph $G[U_i]$.

At each iteration $i < k$, we first compute the set $\mathcal{L}_i$ of maximal pendant paths and maximal internal 
paths of diameter at least $2d+3$, where $d$ is an integer depending on $\epsilon$. 
Then, for every path $\mathcal{P} \in \mathcal{L}_i$ we calculate an independent set $I_{\mathcal{P}}$ 
in the subgraph $G_i[W_{\mathcal{P}} \setminus \Gamma_G[I]]$ induced by those nodes $v$,
whose trees $\mathcal{T}_i(v)$ are contained in $\mathcal{P}$, and which have no neighbors in the set $I$ of nodes that have already been included in the final independent set at the previous iterations. 
Note that by Lemma \ref{lem:helpful-1} graph $G_i[W_{\mathcal{P}} \setminus \Gamma_G[I]]$ is an interval graph.
We compute an independent set $I_{\mathcal{P}}$ by computing an independent set in every maximal 
connected subgraph $H$ of $G_i[W_{\mathcal{P}} \setminus \Gamma_G[I]]$. 
In order to achieve the desired approximation factor, in graphs $H$ of independence number at least $d$ we compute
a $(1+\epsilon/8)$-approximate independent set using a distributed implementation of the algorithm 
for interval graphs given in Section \ref{sec:mis_int}.
In each graph $H$ of independence number less than $d$, we compute an \textit{absorbing} maximum independent set,
i.e., a maximum independent set $I_H$ possessing the property that $|I_H| = \alpha(\Gamma_{G_i}[I_H] \setminus \Gamma_G[I])$.

To explain how such a maximum independent set can be constructed, we first observe that if $\alpha(H) < d$, then nodes of $H$ 
can have neighbors in at most one vertex of $\mathcal{T}_i$ outside $\mathcal{P}$. This is clearly true if $\mathcal{P}$ is pendant. 
Now suppose that $\mathcal{P} = C_1, C_2, \ldots, C_k$  is an internal path connected to $\mathcal{T}_i$ by edges $C_sC_1$
and $C_kC_e$, and assume that $H$ contains a node that has a neighbor in $C_s$, and a node that has a neighbor in $C_e$. Then, since 
$G[C_s \cup C_1 \cup \ldots \cup C_k \cup C_e]$ is an interval graph, it is not hard to see that 
$diam(\mathcal{P}) \leq diam(H) + 4$. Further, notice that the diameter of $H$ is at most $2(d-1)$, as its independence number is
at most $d-1$. Therefore $diam(\mathcal{P}) \leq 2d+2$, which contradicts the assumption that $diam(\mathcal{P}) \geq 2d+3$.

Now if no node of $H$ has a neighbor in a vertex of $\mathcal{T}_i$ outside $\mathcal{P}$, then any maximum 
independent set of $H$ is an absorbing maximum independent set.
If $H$ has a node with a neighbor in a vertex $C$ of $\mathcal{T}_i$ outside $\mathcal{P}$, then it is not hard to see that an
absorbing maximum independent set of $H$ can be obtained by iteratively removing simplicial nodes (and their neighbors) in the order of their remoteness from $C$, i.e., the furthest node is removed first.

At the last iteration $k$, we do everything exactly as in the previous iterations except that $\mathcal{L}_i$ 
is defined to be the set containing all maximal pendant paths of $\mathcal{T}_k$, and all maximal internal paths of independence 
number at least $d$. The reason for this is that we want to use maximal internal paths of large independence number to compute 
the final approximate independent set, but such paths could have small diameter which would not allow us to apply 
Lemma \ref{lem:internal-path-removal} to go easily from $\mathcal{T}_{i}$ to $\mathcal{T}_{i+1}$ in the peeling process. 
Therefore we postpone the processing of paths of large independence number and small diameter until the last iteration,
when it is no longer necessary to update the clique forest for the remaining graph.
Notice that we utilized the lower bound on the diameter of maximal internal paths in the construction of 
an absorbing maximum independent set, and therefore, at the last iteration, we can not use the corresponding arguments.
This is however not a problem, as we do not need the absorption property in the last iteration, and we compute arbitrary 
maximum independent sets for graphs $H$ with $\alpha(H) < d$.

\begin{algfloat}
\normalsize
 \noindent \fbox{\parbox{\textwidth - 10pt}{
 
 \textbf{Input: } A $n$-node chordal graph $G=(V, E)$ with clique forest $\mathcal{T} = (\mathcal{C}, \mathcal{E})$;
 parameter $\epsilon \in (0, 1/2)$.
 
 \vskip2ex
 Let $d = \left\lceil \frac{64}{\epsilon} \right\rceil$, 
 $k = \left\lceil \log \left(\frac{d}{\epsilon} \right) + 2 \right\rceil$,
 $\mathcal{T}_1 = \mathcal{T}$, $U_1 = V$, and $I = \emptyset$. 

 \vskip2ex
\textbf{for} $i=1,2, \dots, k$ \textbf{do}: 
 \begin{enumerate}[leftmargin=0.75cm,nolistsep]
  \item \textbf{Determine Pendant and Internal Paths.}
	\begin{enumerate}[leftmargin=0.5cm,nolistsep]
		\item[] \textbf{if} $i < k$ \textbf{then}
			\begin{enumerate}[leftmargin=0.5cm,nolistsep]
				\item[] Let $\mathcal{L}_i$ be the set that contains all maximal pendant paths of
				$\mathcal{T}_i$, and all maximal internal paths of diameter at least $2d+3$.
			\end{enumerate}
		\item[] \textbf{else} 
			\begin{enumerate}[leftmargin=0.5cm,nolistsep]
				\item[] Let $\mathcal{L}_i$ be the set that contains all maximal pendant paths of
				$\mathcal{T}_i$, and all maximal internal paths of independence number at least $d$.
			\end{enumerate}
		\item[] \textbf{end if} 
	\end{enumerate}
  
  \item \textbf{Compute Independent Set.} \\
	\textbf{for} each path $\mathcal{P} \in \mathcal{L}_i$ \textbf{do} 
  	\begin{enumerate}[leftmargin=0.5cm,nolistsep]
  	\item[] $I_{\mathcal{P}} = \emptyset$.
	\item[] Let $W_{\mathcal{P}} \subseteq U_i$ be the set of nodes $w$ such that $\mathcal{T}_i(w)$ is a subpath of $\mathcal{P}$. 
	\item[] \textbf{for} every maximal connected subgraph $H$ of $G_i[W_{\mathcal{P}} \setminus \Gamma_G[I]]$ \textbf{do}
   		\begin{enumerate}[leftmargin=0.5cm]
   			\item[] \textbf{if} $\alpha(H) < d$ \textbf{then}
   				\begin{enumerate}[leftmargin=0.5cm,nolistsep]
					\item[] Compute an absorbing maximum independent set $I_{H}$ in $H$, if $i < k$, and
					\item[] an arbitrary maximum independent set $I_{H}$ in $H$, if $i = k$.
				\end{enumerate}
   			\item[] \textbf{else}
  				\begin{enumerate}[leftmargin=0.5cm,nolistsep]
					\item[] Compute a $(1+\epsilon/8)$-approximate maximum independent set $I_{H}$ in $H$.
				\end{enumerate}
   			\item[] \textbf{end if}
   			\item[] $I_{\mathcal{P}} = I_{\mathcal{P}} \cup I_{H}$.
   		\end{enumerate}	
  	\end{enumerate}
  	  	
  Let $\displaystyle I_i = \bigcup_{\mathcal{P} \in \mathcal{L}_i} I_{\mathcal{P}}$.

  \item \textbf{Update Data Structures.}
    \begin{enumerate}[leftmargin=1cm]  		
		\item Let $V_i \subseteq U_i$ be such that for each $v \in V_i$, $\mathcal{T}(v)$ is a 
			subpath of one of the paths in $\mathcal{L}_i$. 
		\item Let $U_{i+1} = U_i \setminus V_i$.

		\item \textbf{if} $i < k$ \textbf{then}
			\begin{enumerate}[leftmargin=0.5cm,nolistsep]
				\item[] Let $\mathcal{T}_{i+1}$ be the forest obtained from $\mathcal{T}_i$ 
				by removing all paths in $\mathcal{L}_i$. 
				\item[] As proved in Lemma~\ref{lem:tree-decomp}, $\mathcal{T}_{i+1}$ is the clique forest of $G_{i+1}$.
			\end{enumerate}
		\item Update the independent set: $I = I \cup I_i$.
	\end{enumerate}
 \end{enumerate} 
 \textbf{Return} $I$
}
}
\begin{center}
\algo{alg:misChordal}
 \textbf{Algorithm \ref*{alg:misChordal}.} A deterministic centralized $(1+\epsilon)$-approximation algorithm for 
 the maximum independent set problem in chordal graphs.
\end{center}
\vspace{-0.5cm}
\end{algfloat}

\subsection{Analysis}\label{sec:chAlgAnalysis}
Let $a_i$ be the number vertices of degree at least $3$ in $\mathcal{T}_i$.

\begin{lemma} \label{lem:is-bound-2}
	$\alpha(G_1) \ge a_1$.
\end{lemma}
\begin{proof}
	Let $l_1$ be the number vertices of degree at most 1 in $\mathcal{T}_1$.
	It holds $\alpha(G) \ge l_1 \ge a_1$. The latter inequality follows from the fact that
	the number of vertices of degree at least $3$ in a forest is less than the number of vertices of degree at most 1.
	The former inequality follows from the fact that
	every vertex of degree at most one in $\mathcal{T}_1$ contains at least one node 
	that belongs only to this vertex, and hence the union of these nodes forms an independent set. 
\end{proof}

\begin{lemma} \label{lem:is-bound}
	$\alpha(G_{k+1}) \leq \frac{\epsilon}{2} \alpha(G_1)$.
\end{lemma}
\begin{proof}
	It follows from the proof of Lemma \ref{lem:forestPruning} that $a_k \leq \frac{a_1}{2^{k-1}}$. 
	Hence. using Lemma \ref{lem:is-bound-2} we derive
	\begin{eqnarray} \label{eqn:ak}
		a_k \leq \frac{a_1}{2^{k-1}} \leq \frac{\alpha(G_1)}{2^{k-1}} \leq \frac{\epsilon}{2d} \alpha(G_1).
	\end{eqnarray}
	
	Now, according to Algorithm \ref{alg:misChordal}, $G_{k+1}$ is obtained from $G_k$ by removing nodes $w$
	such that $\mathcal{T}_{k}(w)$ is a subpath of a maximal pendant path or a maximal internal path of independence number
	at least $d$.
	Therefore, every node of $G_{k+1}$ belongs to a vertex in $\mathcal{T}_k$ of degree at least 3,
	or to a vertex of a maximal internal path in $\mathcal{T}_k$ of independence number less than $d$.
	Since there are at most $a_{k}-1$ maximal internal path in $\mathcal{T}_k$, we conclude that
	$\alpha(G_{k+1}) \leq (d-1)(a_{k}-1) + a_k \leq d a_k$, which together with (\ref{eqn:ak}) implies the result.
\end{proof}

\begin{theorem}\label{th:MISapprox}
 For every $\epsilon \in (0, 1/2)$ Algorithm~\ref{alg:misChordal} is a $(1+\epsilon)$-approximation algorithm for the maximum independent set problem on chordal graphs. 
\end{theorem}
\begin{proof}
Let $I^*$ be a maximum independent set in $G_1$. Let $I^* = I^*_{\le k} \cup I^*_{>k}$, where $I_i^* = I^* \cap V_i$, 
$I^*_{\le k} = \bigcup_{i=1}^k I^*_i$, and $I^*_{> k} = I^* \setminus I^*_{\le k}$. 
Since $|I^*_{>k}| \leq \alpha(G_{k+1})$, by Lemma~\ref{lem:is-bound}, we have $|I^*_{>k}| \leq \frac{\epsilon}{2} |I^*|$. 
In the following, we will show that $|I^*_{\le k}| \le (1+\frac{\epsilon}{4})|I|$. This then implies the result, because 
 $$
	|I^*| = |I^*_{\le k}| + |I^*_{>k}| \le \left( 1+\frac{\epsilon}{4} \right)|I| + \frac{\epsilon}{2} |I^*| \ ,  
 $$
 which in turn implies $|I^*| \leq (1+\epsilon) |I|$ (using the assumption $\epsilon < \frac{1}{2}$).

Let  $S_i = \bigcup_{r=1}^{i} \Gamma_{G_{i+1}} [I_r]$.
 In order to show that $|I^*_{\le k}| \le (1+\frac{\epsilon}{4})|I|$ holds, 
 we assign to each binary path $\mathcal{P} \in \mathcal{L}_i$ a subset of nodes 
 $V_{\mathcal{P}} = (W_{\mathcal{P}} \cup \Gamma_{G_i}[I_{\mathcal{P}}]) \setminus S_{i-1}$, if $i < k$, and
 $V_{\mathcal{P}} = W_{\mathcal{P}} \setminus S_{i-1}$, if $i = k$.
It is easy to see that $I^*_{\le k} \subseteq \bigcup_{i=1}^k V_i \subseteq \bigcup_{\mathcal{P} \in \mathcal{L}} V_{\mathcal{P}}$.
It thus remains to show that $\alpha(V_{\mathcal{P}}) \le (1+ \epsilon / 4)  |I_{\mathcal{P}}|$ holds, for every 
$\mathcal{P} \in \bigcup_{i=1}^k \mathcal{L}_i$. 
Let first $i < k$. We distinguish two cases:
 
\begin{enumerate}
	\item $\alpha(W_{\mathcal{P}} \setminus S_{i-1}) < d$. In this case, the algorithm computes an absorbing maximum independent set
	 $I_{\mathcal{P}}$ in $G_i[W_{\mathcal{P}} \setminus S_{i-1}]$, and hence 
	$$
		|I_{\mathcal{P}}| = \alpha(\Gamma_{G_i}[I_{\mathcal{P}}] \setminus S_{i-1}) = 
		\alpha((W_{\mathcal{P}} \cup \Gamma_{G_i}[I_{\mathcal{P}}]) \setminus S_{i-1}) = \alpha(V_\mathcal{P}).
	$$
	
	\item $\alpha(W_{\mathcal{P}} \setminus S_{i-1}) \geq d$. Observe that 
	$\alpha(V_{\mathcal{P}}) = \alpha( (W_{\mathcal{P}} \cup \Gamma_{G_i}[I_{\mathcal{P}}]) \setminus S_{i-1})
		\leq \alpha(W_{\mathcal{P}} \setminus S_{i-1}) + 2.$
		
	By definition of the algorithm, we have:
  	\begin{eqnarray} \label{eqn:inAtLeastD}
		(1+\epsilon / 8) |I_{\mathcal{P}}| & \ge & \alpha(W_{\mathcal{P}} \setminus S_{i-1}) \ge  \alpha(V_{\mathcal{P}}) - 2 \ ,
  	\end{eqnarray}
	which implies $|I_{\mathcal{P}}| \ge \frac{\alpha(W_{\mathcal{P}} \setminus S_{i-1})}{1+\epsilon / 8} \ge \frac{d}{1+\epsilon / 8} \ge 
	\frac{32}{\epsilon}$, and, in particular, $\frac{\epsilon}{8}|I_{\mathcal{P}}| \geq 2$.
	Using the latter in inequality~(\ref{eqn:inAtLeastD}) gives $(1+\epsilon / 4) |I_{\mathcal{P}}| \ge \alpha(V_{\mathcal{P}})$.  
\end{enumerate}

\noindent
For $i=k$ the above analysis becomes simpler, as $V_{\mathcal{P}} = W_{\mathcal{P}} \setminus S_{i-1}$, and we omit the details.
 
\end{proof}

\subsection{Distributed Implementation}\label{sec:chDistImpl}

We will argue now that Algorithm~\ref{alg:misChordal} can be implemented in the {\sf LOCAL} model in $\Order(\frac{1}{\epsilon}\log(\frac{1}{\epsilon})\log^* n)$ rounds.

\begin{theorem}\label{thm:dist-alg-mis}
	For every $\epsilon \in (0, 1/2)$, there is a deterministic $(1+\epsilon)$-approximation algorithm for 
	\textsc{MIS} on chordal graphs that runs in $\Order(\frac{1}{\epsilon}\log(\frac{1}{\epsilon})\log^* n)$ 
	rounds in the {\sf LOCAL} model.
\end{theorem}


Most of the techniques used in this implementation have already been employed in our distributed coloring algorithm, and we therefore omit a lengthy exposition. 

Similar to the coloring algorithm, network nodes obtain local views of the clique forest and execute the peeling process, however, they stop after $k = \Order(\log \frac{1}{\epsilon})$ rounds. 
After each step of the peeling process, we compute independent sets in interval graphs corresponding to the removed 
maximal binary paths of the clique forest.
In general, these interval graphs are disconnected, and we compute independent sets in every connected component in parallel.
For components of small independence number, and therefore of small diameter, we compute the maximum independent sets, 
which are local operations requiring at most $\Order(d)$ rounds. For components of large independence number, we compute 
approximate maximum independent sets in $\Order(\frac{1}{\epsilon} \log^* n)$ rounds using Algorithm~\ref{alg:misInterval}.
The runtime is hence $\Order(\frac{1}{\epsilon}\log(\frac{1}{\epsilon})\log^* n)$.

\section{Lower Bound on the Round Complexity for MIS}
\label{app:lowerMIS}

In this section we show that any randomized $(1 + \epsilon)$-approximation algorithm for \textsc{MIS} 
in the {\sf LOCAL} model requires $\Omega(\frac{1}{\epsilon})$ rounds, even on paths.

Let $P_n=(V, E)$ be the path on $n$ nodes with $V = \{v_1, v_2, \dots, v_n\}$ and 
 $E = \{ v_i v_j ~|~ |i-j| = 1 \}$. We assume that the path is labelled, i.e., every node $v_i$
is assigned a unique label $\ell(v_i)$, where $\ell$ is chosen uniformly at random from the set of bijections
between $V$ and $\{1, 2, \dots, n\}$.
In the following proof, for $i \le j$ we use the notation $V_{i,j} := \{v_i, v_{i+1}, \dots, v_j \}$.
 
\begin{theorem} \label{thm:lb-mis}
 For every $\epsilon > 0$ and $n$ large enough, every randomized algorithm in the {\sf LOCAL} model with expected approximation factor at most $1+ \epsilon$ for \textsc{MIS} requires $\Omega(\frac{1}{\epsilon})$ rounds.
\end{theorem}
\begin{proof}
Let $\mathbf{A}$ be an $r$-round distributed randomized algorithm for the maximum independent set problem with 
expected approximation ratio at most $1+ \epsilon$. Let $I$ be the output independent set computed by $\mathbf{A}$ on $P_n$. We define 
$p_i := \Pr \left[ v_i \in I \right]$, where the probability is taken over the random bits of the algorithm and the 
labelling function. Then, by linearity of expectation, $\Exp |I| = \sum_{i \in [n]} p_i$. Since the size of a maximum 
independent set in $P_n$ is $\lceil n/2 \rceil$, and the expected approximation factor of $\mathbf{A}$ is at most $1+\epsilon$, we have 
\begin{eqnarray}\label{eqn:392}
(1+\epsilon) \cdot \sum_{i \in [n]} p_i \ge \lceil n/2 \rceil. 
\end{eqnarray}

Next, notice that in $r$ rounds, every node can only learn its local $r$-neighborhood. Hence, by symmetry, all nodes
that are at distance at least $r+1$ from the boundary of the path (i.e., from $v_1$ and $v_n$) have the same probability $p$
to be chosen into the independent set, i.e., $p_i = p$ for every $i \in \{ r+2, r+3, \ldots, n-r-1 \}$. 
Since a maximum independent set in $G[V_{r+2, n-r-1} ]$ is of size at most 
$\frac{n-2r-2}{2} + 1$, we have $p \le \frac{1}{2} + \Order(\frac{1}{n})$.  

For every $i \in \{ r+2, r+3, \ldots, n - 3 r-3 \}$ we denote by $X_i$ the number of nodes of $V_{i, i+2r+2}$ selected into $I$,
i.e., $X_i = |V_{i, i+2r+2} \cap I|$. We will argue now that $X_i \leq r+5/4 + \Order(\frac{1}{n})$. 
To this end, suppose first that the event $v_i \in I$ happens 
(which happens with probability $p$). Since the $r$-neighborhoods of $v_i$ and $v_{i+2r+1}$ are disjoint, and the $r$-neighborhoods 
of $v_i$ and $v_{i+2r+2}$ are disjoint too, it is equally likely that $v_{i+2r+1}$ or $v_{i+2r+2}$ will be selected into the 
independent set. Therefore, $\Pr \left[ v_{i+2r+2} \in I \, | \, v_i \in I \right] \le \frac{1}{2}$. Notice that if $v_{i+2r+2} \notin I$, 
then $X_i \le r+1$. Hence, we have $\Exp[X_i | v_i \in I] \leq r+1 + 1/2$. 
On the other hand if $v_i \notin I$, then $X_i \le r+1$ always holds. Thus, 
\begin{eqnarray*}
\Exp[X_i] = \Exp[X_i | v_i \in I] \cdot \Pr[v_i \in I] + \Exp[X_i | v_i \notin I] \cdot \Pr[v_i \notin I] \leq \\
(r+1+1/2)p + (r+1)(1-p) \leq r + 5/4 + \Order(\frac{1}{n}).
\end{eqnarray*}
This proves the claim. Since $\Exp[X_i] = \sum_{j=i}^{i+2r+2} p_i = p (2r+3)$, we obtain 
$p \le \frac{r+5/4 + \Order(\frac{1}{n})}{2r+3}$. Using this in inequality~(\ref{eqn:392}),
we obtain: 
\begin{eqnarray*}
\lceil n/2 \rceil \le (1+\epsilon) \cdot \sum_{i \in [n]} p_i \le (1+\epsilon) \left( (2r+2) + (n-2r-2) p \right) \le (1+\epsilon) n \left(\frac{1}{2} - \frac{1}{8r+12} \right) + \Order(1)  \ ,
\end{eqnarray*}
which in turn implies $r = \Omega(\frac{1}{\epsilon})$ and proves the theorem.
\end{proof}

\section{Conclusion} \label{sec:conclusion}

In this paper, we gave distributed $(1+\epsilon)$-approximation algorithms for \textsc{MVC} and \textsc{MIS} on chordal graphs. 
We showed that in chordal graphs network nodes can obtain coherent views of a global tree decomposition, which enabled us to
exploit the tree structure of the input graph for the design of algorithms. 

How can we extend the class of graphs on which we can solve \textsc{MVC} and \textsc{MIS} within a small approximation factor even 
further? In particular, how can we handle graphs that contain longer induced cycles, such as $k$-chordal graphs (for some integer $k$)?

\bibliographystyle{plainurl}
\bibliography{chordal}

\end{document}